\documentclass[a4paper,11pt]{article}
\usepackage{amssymb,amsmath,latexsym,amsthm}
\usepackage{color,mathrsfs}
\usepackage{url}
\usepackage{enumerate}
\usepackage[square,sort,comma,numbers]{natbib}
\usepackage{dsfont}
\usepackage{diagbox}
 \usepackage{geometry}
\usepackage{graphicx}
\usepackage{colortbl}

\usepackage{endnotes}

\renewcommand{\textbf}[1]{\begingroup\bfseries\mathversion{bold}#1\endgroup}

\usepackage{fancyhdr}
\usepackage{pstricks,pst-plot,pst-node,pstricks-add}

\newlength{\bibitemsep}
\newlength{\bibparskip}
\let\oldthebibliography\thebibliography
\renewcommand\thebibliography[1]{%
  \oldthebibliography{#1}%
  \setlength{\parskip}{\bibitemsep}%
  \setlength{\itemsep}{\bibparskip}%
}

\newtheorem{thm}{Theorem}[section]
\newtheorem{defi}{Definition}[section]
\newtheorem{corollary}[thm]{Corollary}
\newtheorem{prop}[thm]{Proposition}
\newtheorem{Conjecture}[thm]{Conjecture}
\newtheorem{lemma}[thm]{Lemma}

\theoremstyle{definition}
\newtheorem{remark}[thm]{Remark}

\newtheorem{example}[thm]{Example}

\newcommand{\argmin}{\mathop{\rm argmin}\nolimits}
\newcommand{\hcp}{\mathop{\rm hcp}\nolimits}

\newcommand{\R}{\mathbb R}

\newcommand{\Z}{\mathbb Z}

\numberwithin{equation}{section}

\def\XXint#1#2#3{{\setbox0=\hbox{$#1{#2#3}{\int}$}
    \vcenter{\hbox{$#2#3$}}\kern-.5\wd0}}

\allowdisplaybreaks
\date{date}
\begin{document}
\title{Minimizing lattice structures for Morse potential energy \\ in two and three dimensions}
\author{Laurent B\'{e}termin\\ \\
QMATH, Department of Mathematical Sciences, University of Copenhagen,\\ Universitetsparken 5, DK-2100 Copenhagen \O, Denmark.\\ \texttt{betermin@math.ku.dk}. ORCID id: 0000-0003-4070-3344 }
\date\today
\maketitle

\begin{abstract}
We investigate the local and global optimality of the triangular, square, simple cubic, face-centred-cubic (FCC), body-centred-cubic (BCC) lattices and the hexagonal-close-packing (HCP) structure for a potential energy per point generated by a Morse potential with parameters $(\alpha,r_0)$. In dimension 2 and for $\alpha$ large enough, the optimality of the triangular lattice is shown at fixed densities belonging to an explicit interval, using a method based on lattice theta function properties. Furthermore, this energy per point is numerically studied among all two-dimensional Bravais lattices with respect to their density. The behaviour of the minimizer, when the density varies, matches with the one that has been already observed for the Lennard-Jones potential, confirming a conjecture we have previously stated for differences of completely monotone functions. Furthermore, in dimension 3, the local minimality of the cubic, FCC and BCC lattices are checked, showing several interesting similarities with the Lennard-Jones potential case. We also show that the square, triangular, cubic, FCC and BCC lattices are the only Bravais lattices in dimensions 2 and 3 being critical points of a large class of lattice energies (including the one studied in this paper) in some open intervals of densities, as we observe for the Lennard-Jones and the Morse potential lattice energies. More surprisingly, in the Morse potential case, we numerically found a transition of the global minimizer from BCC, FCC to HCP, as $\alpha$ increases, that we partially and heuristically explain from the lattice theta functions properties. Thus, it allows us to state a conjecture about the global minimizer of the Morse lattice energy with respect to the value of $\alpha$. Finally, we compare the values of $\alpha$ found experimentally for metals and rare-gas crystals with the expected lattice ground-state structure given by our numerical investigation/conjecture. Only in a few cases does the known ground-state crystal structure match the minimizer we find for the expected value of $\alpha$. Our conclusion is that the pairwise interaction model with Morse potential and fixed $\alpha$ is not adapted to describe metals and rare-gas crystals if we want to take into consideration that the lattice structure we find in nature is the ground-state of the associated potential energy.
\end{abstract}


\tableofcontents

\medskip

\section{Introduction and main results}

A fundamental question in Mathematical Physics that has been actively investigated recently is the following ``Crystal Problem" (also called ``The crystallization conjecture" see e.g. \cite{RadinLowT,BlancLewin-2015}): Why are solids crystalline? Answering this question in a rigorous mathematical way is known to be extremely challenging, even though the interactions between atoms or molecules are assumed to be a sum of pairwise potentials. Whereas the one-dimensional version of this problem is well-understood \cite{Ventevogel1,VN2, VN3, Rad1,Crystbinary1d}, only few results have been proved in dimensions $2$ and $3$ for models consisting of short-range interactions \cite{Rad2,Stef1,Stef2,Luca:2016aa,Friedrich:2018aa}, perturbations of the hard-sphere potential \cite{Crystal,ELi,TheilFlatley} and oscillating functions \cite{Suto1}.

\medskip

In 1929, Morse \cite{Morse} solved the three-dimensional Schr\"odinger equation with potential
$$
V_M(|x|):=e^{-2\alpha(|x|-r_0)}-2e^{-\alpha(|x|-r_0)},\quad x\in \R^d,
$$
known now as the ``Morse Potential", where $|.|$ is the euclidean norm in $\R^d$. This is an attractive-repulsive potential (see Figure \ref{fig:Morseplots} for a plot), i.e. a decreasing-increasing potential having one well. Parameters $r_0$ and $\alpha$ respectively represent the minimizer of $V_M$ and the hardness of the interaction (as $\alpha$ increases, $V_M$ goes to a hard-core potential, see Figure \ref{fig:Morseplots}). This potential is known to be a canonical model for social aggregation -- e.g. swarming and flocking -- as explained in \cite[Sect. 4]{PrimerSwarm} (see also \cite{Selfpropelled,CompactGlobMinMorse} and references therein). Furthermore, it has been shown (see e.g. \cite{RareGasHorton}) that $V_M$ provides a description of the vibrational properties of rare-gas crystals which is better than the one given by the quantum harmonic oscillator (see also \cite{Raff1990,AlimietalMorse,ParsonMorse,BarkerMorse}). Moreover, interactions in cubic metals are also well-described by the Morse potential, as explained in \cite{GirifalcoWeizer,LincolnetalMorse,SharmaKachhavaMorse,MilsteinBCC,PamuketalMorse,HungetalMorse} and \cite[p. 22]{KuboNagamiya}. The values of parameters $(\alpha,r_0)$ can then be computed from experimental data for many metals and rare-gas crystals.

\begin{figure}[!h]
\centering
\includegraphics[width=7cm]{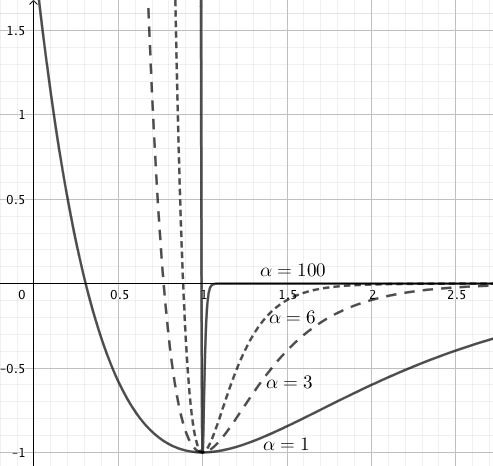} 
\caption{Plots of $V_M$ for $r_0=1$ and $\alpha\in \{1,3,6,100\}$.}
\label{fig:Morseplots}
\end{figure}

\medskip

Crystallization problems for Morse potential have not received so much attention. Ventevogel and Nijboer \cite[p. 276]{VN2} pointed out the fact that, even in dimension 1, any general crystallization result is difficult to reach for $V_M$ (they claim to have the solution among one-dimensional periodic configurations with 2 and 3 points in each period). However, in a recent paper, Bandegi and Shirokoff \cite[Sect. 6.1]{Bandegi:2015aa} give numerical evidences for the global optimality of the equidistant configuration for some values of the density and the parameters of the Morse potential, using convex relaxation arguments. In higher dimension, no such result exists and only local stability properties has been proved in \cite{MilsteinMorse1970}, using Born's stability criteria for crystals, following the methods developed in \cite{born1940,misra1940}.

\medskip

Instead of investigating the pure Crystal Problem involving $V_M$ -- i.e. minimizing the interaction energy among all the possible point configurations -- we choose to study the minimization of the energy per point among periodic lattices where the points are interacting via a Morse potential. This point of view has been taken in several previous works in Number Theory \cite{Rankin,Eno2,Ennola,Cassels,Diananda,Mont,SarStromb, Coulangeon:kx,Coulangeon:2010uq,Gruber}, optimal point configurations problems \cite{CohnKumar} and Mathematical Physics \cite{NonnenVoros,AftBN,CheOshita,Sandier_Serfaty,BeterminPetrache,BetKnupfdiffuse}. It is a natural first step for keeping or rejecting periodic structures which could be good candidates for the Crystal Problem associated to the interaction potential. We have already studied the Lennard-Jones type potential case in \cite{Betermin:2014fy,BetTheta15,Beterminlocal3d,Beterloc,OptinonCM} and our goal is to give the same kind of quantitative results for the Morse potential.

\medskip

Let us first define our spaces of periodic lattices. Let $\mathcal{L}_d^\circ(V)$ be the space of Bravais lattices, i.e. having the form $L=\bigoplus_{i=1}^d \Z u_i$ where $\{u_i\}_i$ is a basis of $\R^d$, of area (in dimension $2$, and we will note it $A$) or volume $V>0$ (i.e. the volume  $|\det \{u_i\}_i|$ of their unit cell) and $\mathcal{L}_d$ be the space of all $d$-dimensional Bravais lattices. We also write $\mathcal{P}_d$ the space of all $d$-dimensional periodic configurations, i.e. all the possible finite unions of Bravais lattices. We hence have, for any volume $V$, $\mathcal{L}_d^{\circ}(V)\subset \mathcal{L}_d\subset \mathcal{P}_d$.  Furthermore, in dimension $2$ and $3$, as explained for instance in \cite{Terras_1988}, any Bravais lattice $L\in \mathcal{L}_d$ can be parametrized by a vector $(x,y,A)$ or $(u,v,x,y,z,V)$ where $(x,y)\in \R^2$ (resp. $(u,v,x,y,z)\in \R^5$) belongs to a fundamental domain containing only one copy of each lattice (it is basucally due to the reduction of quadratic forms). This parametrization will be used in Sections \ref{sec-numeric2d} and \ref{sec-numeric3d}. In particular, for any $E:\mathcal{L}_d\to \R$ of class $C^2$, the gradient and the Hessian of $E$ at $L\in \mathcal{L}_d$ will be respectively denoted by $\nabla_L E[L]$ and $D^2 E[L]$. For more details, see \cite{Beterloc,Beterminlocal3d}. Furthermore, the same differentation with respect to the structure can be done for periodic configurations (see e.g. \cite{Coulangeon:2010uq} for details).

\medskip

We now define the energy we want to focus on. Writing 
$$
V_M(|x|)=e^{\alpha r_0}\left(e^{\alpha r_0}e^{-2\alpha |x|} -2e^{-\alpha |x|} \right)=:e^{\alpha r_0} f(|x|^2),\quad f(r):=e^{\alpha r_0} e^{-2\alpha \sqrt{r}}-2e^{-\alpha \sqrt{r}},
$$
the goal of this paper is to investigate rigorously and numerically the energy per point defined by
\begin{equation}
E_{\alpha,r_0}[L]:=\sum_{p\in L} f(|p|^2)=e^{\alpha r_0} \sum_{p\in L} e^{-2\alpha |p|}-2\sum_{p\in L} e^{-\alpha |p|},
\end{equation}
among Bravais lattices $L\in \mathcal{L}_d$ (or among periodic configurations). This energy, as it is the case for Lennard-Jones type potentials and for the difference of Yukawa potentials (see \cite[Sec. 5]{BetTheta15}), can be seen as a difference of competitive interactions with completely monotone potentials (i.e. the functions are positive and the signs of their derivatives alternate). It has been shown in \cite[Sec. 3.1]{BetTheta15} that $L\mapsto \sum_{p\in L} e^{-\beta |p|}$ is minimized by the triangular lattice and the square lattice $\Z^2$ is a saddle point of the energy, in $\mathcal{L}_2^\circ(A)$ for all fixed $A$ and for all $\beta>0$, following the lattice theta function properties proved in \cite{Mont}. Thus, studying the minimization on $\mathcal{L}_d^\circ(V)$ for a difference of such lattice energies must show, as in the Lennard-Jones case, many other minimizing lattices with respect to $V$, at least in dimension 2 (see \cite{Beterloc}).

\medskip

Using a method we have developed in \cite{Betermin:2014fy,BetTheta15}, we find an interval of areas $A$ such that the triangular lattice $\Lambda_A$ (see Figure \ref{fig-lattice2d}) defined by
\begin{equation}
\Lambda_A:=\sqrt{\frac{2A}{\sqrt{3}}}\left[\Z\left(1,0 \right)\oplus \Z\left( \frac{1}{2},\frac{\sqrt{3}}{2} \right)  \right]
\end{equation}
is the unique minimizer, up to rotation, of $E_{\alpha,r_0}$ in $\mathcal{L}_2^\circ(A)$, when $\alpha$ is not too small. Furthermore, the non-optimality of the triangular lattice at low density (i.e. large area) can be shown as a direct consequence of \cite[Thm 1.5]{OptinonCM}, which is itself a consequence of Montgomery Theorem \cite[Thm 1]{Mont} and the functional equation for the lattice theta function \cite[Eq. (2)]{Mont} defined by
\begin{equation}\label{def-thetafunction}
\theta_L(\alpha):=\sum_{p\in L} e^{-\pi \alpha |p|^2}.
\end{equation}
\begin{figure}[!h]
\centering
\includegraphics[width=6cm]{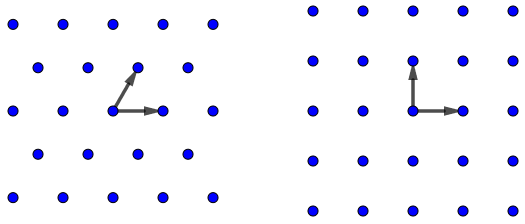} 
 \caption{Triangular lattice $\Lambda_A$ and square lattice $\sqrt{A}\Z^2$. The area of their unit cells is $A$} \label{fig-lattice2d}
\end{figure}
More precisely, we prove the following result.

\begin{thm}[Sufficient conditions for the optimality of $\Lambda_A$]\label{thm-main}
Let $\alpha>0$. If $A$ satisfies one of the following conditions:
\begin{enumerate}
\item[(C1)] $\frac{8\pi}{\alpha^2}\leq A<\frac{\pi r_0}{\alpha}$ and $e^{\alpha r_0} e^{-\frac{\alpha^2 A}{\pi}}-1 \geq e^{-\frac{\alpha^2 A}{4\pi}}$,
\item[(C2)] $A<\min \left\{ \frac{\pi r_0}{\alpha}, \frac{8\pi}{\alpha^2} \right\}$ and $\displaystyle e^{\alpha r_0} e^{-\frac{\alpha^2 A}{\pi}}-1\geq \frac{64\pi^2}{e^2 A^2\alpha^4}$,
\end{enumerate}
then the triangular lattice $\Lambda_A$ is the unique minimizer of $E_{\alpha,r_0}$ in $\mathcal{L}_2^\circ(A)$, up to rotation.

\medskip

Furthermore, for any $\alpha,r_0\in (0,+\infty)$ there exists $A_1$ such that for any $A>A_1$, $\Lambda_A$ is not a minimizer of $E_{\alpha,r_0}$ in $\mathcal{L}_2^\circ(A)$.
\end{thm}
\begin{remark}
If $\alpha=6$ and $r_0=1$, then $\min \left\{ \frac{\pi r_0}{\alpha}, \frac{8\pi}{\alpha^2} \right\}=\frac{\pi}{6}\approx 0.524$, $(C2)$ is satisfied and, solving numerically $e^{\alpha r_0} e^{-\frac{\alpha^2 A}{\pi}}-1\geq \frac{64\pi^2}{e^2 A^2\alpha^4}$, we find that $0.0139\leq A\leq 0.5034$ which is then an interval of area where the triangular lattice is the minimizer of $E_{6,1}$ in $\mathcal{L}_2^\circ(A)$. Other numerical values are given in Table \ref{table-C1C2}. \\
In particular, we notice that this method does not apply for $\alpha<3.078$ (approximatively), where none of the conditions are satisfied. This is due to the lower bound \eqref{Eq:LBuA} we have used in the proof. Numerical investigations show that area bounds exist for smaller value of $\alpha$, but we have chosen conditions $(C1)$ and $(C2)$ because they are simple and tractable and we are mostly interested in the $\alpha=6$, $r_0=1$ case which is comparable to the classical Lennard-Jones potential case \eqref{def-LJ} (see Figure \ref{fig-LJMorse}). Furthermore, in most of the case (for instance in \cite{GirifalcoWeizer,LincolnetalMorse,SharmaKachhavaMorse,PamuketalMorse,HungetalMorse,Raff1990,AlimietalMorse,ParsonMorse,BarkerMorse}), the values of $\alpha$ (rescaled such that $r_0=1$) computed from experimental data are larger than $3.078$.
\end{remark}

Since both conditions $(C1)$ and $(C2)$ are satisfied for large values of $\alpha$, we then derive the following consequence of Theorem \ref{thm-main} that gives an explicit interval of areas where the triangular lattice is minimal.

\begin{corollary}[Interval of areas for the optimality of the triangular lattice when $\alpha$ is large]\label{thm-triang}
Let $A_0=-\frac{4\pi}{\alpha^2}\log X_0$ where $X_0$ is the unique solution of $e^{\alpha r_0}X^4-X-1=0$ on $[e^{-\alpha r_0/4},e^{-2}]$.
If $\alpha> \frac{8+\log 2}{r_0}$ and $A$ satisfies
$$
\frac{8\pi}{\alpha^2 \sqrt{e^{\alpha r_0-8}-1}}\leq A\leq A_0<\frac{\pi r_0}{\alpha},
$$
then $\Lambda_A$ is the unique minimizer of $E_{\alpha,r_0}$, up to rotation, in $\mathcal{L}_2^\circ(A)$.\\
\end{corollary}
\begin{remark}
We notice that, by assumption, $A_0\geq \frac{8\pi}{\alpha^2}$, and therefore the optimality of the triangular lattice at very high density is not proved. We also remark that it is possible to reach any small value of $A$ by increasing $\alpha$ (or $r_0$) and such that $\Lambda_A$ is minimal in $\mathcal{L}_2^\circ(A)$ for $E_{\alpha,r_0}$.
\end{remark}

This result ensures the possibility to get a triangular lattice as a ground-state at a certain density (and in the neighborhood of it), using a Morse potential. However, it also shows the limits of our methods, as discussed in Section \ref{sec-limit}. In particular, we cannot get the optimality of the triangular lattice in the high density limit $A\to 0$ as it was the case for the Lennard-Jones type potentials in \cite{BetTheta15}, and as we have numerically checked it for $V_M$. A modification of the Morse potential is then proposed in Section \ref{sec-limit} in order to fill this gap. Furthermore, we also remark that the whole method developed in \cite{BetTheta15} cannot be applied for the Morse potential in order to show the global optimality, in $\mathcal{L}_2$, of a triangular lattice for $E_{\alpha,r_0}$. Indeed, it is straightforward to prove that any global minimizer of $E_{\alpha,r_0}$ on $\mathcal{L}_2$ has an area smaller than $r_0^2$, and it is not possible -- using our method -- to show that $\Lambda_A$ is the unique minimizer of $E_{\alpha,r_0}$ for any $A\in (0,r_0^2]$, whatever $\alpha$ and $r_0$ are.

\medskip

We have performed a numerical investigation of the minimizers of $E_{\alpha,r_0}$ in $\mathcal{L}_2^\circ(A)$ when $A$ varies, as it was done for the classical  Lennard-Jones potential
\begin{equation}\label{def-LJ}
V_{LJ}(r)=\frac{1}{r^{12}}-\frac{2}{r^6}
\end{equation}
in \cite{Beterloc}. Contrary to the latter, the lack of homogeneity of $V_M$ makes the systematic analysis of the local extrema of $E_{\alpha,r_0}$ very difficult, and we cannot find explicit analytic bounds for the local optimality of $\Lambda_A$ and the square lattice $\sqrt{A}\Z^2$ depicted in Figure \ref{fig-lattice2d}. We only propose, in Lemma \ref{lem:asympttriangular}, an asymptotic result for the local minimality of $\Lambda_A$ for small values of $A$. The details of our numerical study can be seen in Section \ref{sec-numeric2d}, based in particular on the minimization among rectangular lattices $\sqrt{A} L_y$  (i.e. its unit cell is a rectangle with sides of length $y^{\pm 1/2}$) and rhombic lattices $\sqrt{A}L_\theta$ (i.e. its unit cell is a rhombus with smallest angle $\theta$) where $L_y,L_\theta\in \mathcal{L}_2^\circ(1)$ are respectively defined by
\begin{align}
 &L_y:=\Z\left(\frac{1}{\sqrt{y}},0  \right)\oplus \Z \left(0,\sqrt{y} \right)\quad y\geq 1;\label{rectangular}\\
& L_\theta=\Z u_\theta \oplus \Z v_\theta,  \quad |u_\theta|=|v_\theta|, \quad (\widehat{u_\theta,v_\theta})=\theta, \quad \frac{\pi}{3}\leq \theta\leq \frac{\pi}{2}. \label{rhombic}
\end{align}
The two most important observations are the following:
\begin{enumerate}
\item Global minimizer: for all $\alpha,r_0$, the global minimizer of $E_{\alpha,r_0}$ in $\mathcal{L}_2$ seems to be a triangular lattice, as it seems to be the case for Lennard-Jones type potentials (see \cite{BetTheta15,Beterloc,OptinonCM}).
\item Confirmation of our conjecture: the minimizer's transition with respect to $A$ follows the same law as the one we have observed for the classical Lennard-Jones potential in \cite{Beterloc} This supports a conjecture we have stated in \cite[Sec. 5.4]{Beterloc} and where the same phenomenon is expected for all difference of completely monotone functions having only one well. More precisely, as $A$ increases, the evolution of the minimizer of $E_{\alpha,r_0}$ in $\mathcal{L}_2^\circ(A)$ is depicted in Table \ref{table-conj}, where the case $\alpha=6,r_0=1$ has been chosen, and compared to the classical Lennard-Jones case (in such a way that $f''(1)=V_{LJ}''(1)$ and both have their absolute minimum for $r=1$). We observe a transition triangular-rhombic-square-rectangular for the minimizer of $E_{\alpha,r_0}$ in $\mathcal{L}_2^\circ(A)$, as $A$ increases, which seems to stay true for other values of $(\alpha,r_0)$ (like $\alpha=r_0=3$ for instance). As $A\to +\infty$, the minimizer is rectangular and becomes more and more thin. More details are given in Section \ref{subsec-Conj}. This minimizer's behaviour is also similar to the one appearing in the two-components Bose-Einstein Condensates described by Ho and Mueller in \cite[Fig. 1 and 2]{Mueller:2002aa}. The same phenomenon is also naturally expected in other physical and biological models involving infinite lattices and competitive interactions.
\end{enumerate}

\begin{table}\label{table-conj}
\begin{center}
\begin{tabular}[H]{|c|c|}
\hline
\textbf{Area $A$} & \textbf{Minimizer in $\mathcal{L}_2^\circ(A)$}\\
\hline
(M): $A< 1.1560011044$ & Triangular\\
(LJ): $A< 1.138$ & \\
\hline
(M): $1.1560011044<A<1.1560011045$ & Rhombic\\
(LJ): $1.138<A<1.143$ & \\
\hline
(M): $1.1560011045<A<1.291$ & Square\\
(LJ): $1.143<A<1.268$ & \\
\hline
(M): $A>1.291$ & Rectangular\\
(LJ): $A>1.268$ & (more and more thin as $A\to +\infty$)\\
\hline
\end{tabular}
\caption{Shape of the minimizer with respect to the area $A$ (numerical values) for Morse with $\alpha=6$,$r_0=1$ and Lennard-Jones potential \eqref{def-LJ} for which the values are taken from \cite{Beterloc}.}
\end{center}
\end{table}

In particular, we notice that the triangular and square lattices are the only one staying critical points of $E_f$ in $\mathcal{L}_2^\circ(A)$ for $A$ in some open intervals. We call these lattices `volume-stationary' and we have also remarked this phenomenon in \cite{Beterloc} for the classical Lennard-Jones potential. The next result, based on Gruber's result \cite[Cor 2]{Gruber}, shows that this property has a certain universality. 

\begin{thm}[Volume-stationary lattices in dimension $2$]\label{prop:stat2d}
Let $d=2$ and $f:(0,+\infty)\to \R$ be a nonzero function such that
\begin{enumerate}
\item as $r\to +\infty$, we have $f(r)=O(r^{-d/2-\eta})$ for some $\eta>0$,
\item for any $r>0$, it holds $\displaystyle f(r)=\int_0^{+\infty} e^{-rt} d\mu_f(t)$ for some Borel measure $\mu_f$ on $(0,+\infty)$.
\end{enumerate}
For any $L\in \mathcal{L}_d$, we define 
$$
E_f[L]:=\sum_{p\in L\backslash \{0\}} f(|p|^2),
$$
and let $L_0\in \mathcal{L}_d^\circ(1)$. There exists an open interval $I$ such that $\nabla_{L_0} E_f[\sqrt{A} L_0]=0$ for all $A\in I$ if and only if $L_0\in \{\Z^2,\Lambda_1\}$, thus $I=(0,+\infty)$.
\end{thm}
\begin{remark}
The same result is stated in dimension $d=3$ in Theorem \ref{prop:stat3d}, and a more general result (see Theorem \ref{thm:eutactic}) will be proved for $d$-dimensional lattice energies in Section \ref{sec:stat}: all the layers of such lattices $L_0$ are strongly eutactic in the sense of Definition \ref{def:eutactic}. Note that the Morse potential $f$ satisfies the assumption of the theorem because $\mu_f$ is defined by \eqref{eq:laplacetransformMorse} (as well as any Lennard-Jones type potential), thus the result is actually true for $E_{\alpha,r_0}$ (and also for the Lennard-Jones energy studied in \cite{Beterloc}). 
\end{remark}

In dimension 3, the problem is obviously more difficult. The space $\mathcal{L}_3^\circ(V)$ of Bravais lattices with fixed volume $V$ is a five-dimensional space and only local optimality results have been proved for usual interaction potentials \cite{Ennola,BeterminPetrache,Beterminlocal3d}, even for the completely monotone potentials (see e.g. \cite{BetSoftTheta} for a review in the soft lattice theta function case). The only global minimality results is proved by Sarnak and Str\"ombergsson in \cite{SarStromb} for the height of the three-dimensional torus (i.e. the derivative of the Epstein zeta function at the origin), using a computer assisted proof. The exact formula for the partial derivatives of $E_{\alpha,r_0}$ are known (see e.g. \cite{Beterminlocal3d}), but their systematic analysis is again very difficult, by lack of homogeneity of $f$ which contains exponential terms. The four important (for us, in this context) three-dimensional lattices of unit density are the simple cubic lattice $\Z^3\in \mathcal{L}^\circ_3(1)$, Face-Centred-Cubic (FCC) lattice $\mathsf{D}_3\in \mathcal{L}^\circ_3(1)$, Body-Centred Cubic (BCC) lattice $\mathsf{D}_3^*\in \mathcal{L}^\circ_3(1)$ and the Hexagonal-Close-Packing (HCP) structure $\hcp\in \mathcal{P}_3\backslash \mathcal{L}_3$ depicted in Figure \ref{fig-latticethreedimensions}, and defined by
\begin{align}
 &\Z^3=\Z(1,0,0)\oplus \Z(0,1,0)\oplus \Z(0,0,1);\label{cubic}\\
 &\mathsf{D}_3:=2^{-\frac{1}{3}}\left[\Z(1,0,1)\oplus \Z(0,1,1)\oplus \Z(1,1,0)  \right];\label{FCC}\\
 &\mathsf{D}_3^*:=2^{\frac{1}{3}}\left[\Z(1,0,0)\oplus \Z(0,1,0)\oplus \Z\left(\frac{1}{2},\frac{1}{2},\frac{1}{2}  \right)  \right]; \label{BCC}\\
&\hcp:= L\cup \left( L+\left(\frac{1}{2},\frac{1}{\sqrt{12}},\sqrt{\frac{2}{3}}  \right) \right),\quad  L:=\Z(1,0,0)\oplus \Z\left(\frac{1}{2},\frac{\sqrt{3}}{2},0 \right)\oplus \Z\left( 0,0,\sqrt{\frac{8}{3}} \right).\label{HCP}
\end{align}
\begin{figure}[!h]
\centering
\includegraphics[width=3cm]{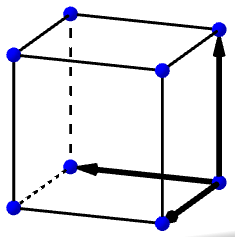} \quad \includegraphics[width=3cm]{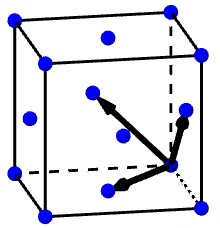}\quad \includegraphics[width=3cm]{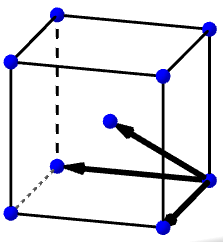}\quad \includegraphics[width=4cm]{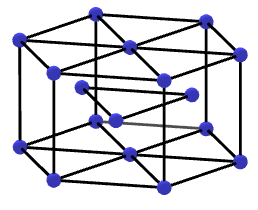}
 \caption{Three-dimensional periodic lattices. The cubic lattice $\Z^3$, the FCC lattice $\mathsf{D}_3$, the BCC lattice $\mathsf{D}_3^*$ and the HCP structure $\hcp$.}\label{fig-latticethreedimensions}
\end{figure}

Using the formula proved in \cite{Beterminlocal3d}, we have numerically studied the local optimality of the cubic lattices $\Z^3$, $\mathsf{D}_3$ and $\mathsf{D}_3^*$ for $E_f$ in $\mathcal{L}_3^\circ(V)$, when $V$ varies. The results, presented in Table \ref{table-3dcubic}, are similar with the one found for the classical Lennard-Jones potential in \cite{Beterminlocal3d}. In particular, we also observe that these cubic lattices are the only one staying critical points for $E_f$ in some open intervals of volumes $V$. It turns out that this phenomenon, as the one observed in dimension $2$ for $\Z^2$ and $\Lambda_1$ and proved in Theorem \ref{prop:stat2d}, is also universal as shown in the next result, again based on Gruber's result \cite[Cor 2]{Gruber}, and generalized in Section \ref{sec:stat}.

\begin{thm}[Volume-stationary lattices in dimension $3$]\label{prop:stat3d}
Let $f$ and $E_f$ be defined as in Theorem \ref{prop:stat2d} with $d=3$. A Bravais lattice $L_0\in \mathcal{L}_3^\circ(1)$ satisfies $\nabla_{L_0} E_f[V^{1/3} L_0]=0$ for all $V\in I$ where $I$ is an open interval if and only if $L_0\in \{\Z^3,\mathsf{D}_3,\mathsf{D_3}^*\}$, thus $I=(0,+\infty)$.
\end{thm}

\medskip

According to \cite{ModifMorse}, it is important to notice that the global minimum for the three-dimensional Morse lattice energy in $\mathcal{P}_3$ should be an HCP structure, for large parameter $\alpha$, which does not belong to the class of Bravais lattice we are interested in. This also holds for the Lennard-Jones potential for large exponents (see e.g. \cite{ModifMorse,OptinonCM}). Even though $\mathcal{P}_3\ni \hcp\not \in \mathcal{L}_3$, using formula \eqref{HCP} we have computed the numerical values of $E_{\alpha,r_0}[\lambda \hcp]$.

\medskip

Surprisingly, we numerically found the following transition of minimizers. Defining 
$$
H_\alpha:=\min_\lambda E_{\alpha,1}[\lambda \hcp],\quad B_\alpha:=\min_\lambda E_{\alpha,1}[\lambda \mathsf{D}_3^*],\quad F_\alpha:=\min_\lambda E_{\alpha,1}[\lambda \mathsf{D}_3],$$ 
we obtain:
\begin{enumerate}
\item If $\alpha\leq \alpha_0$, where $\alpha_0\in (3.05,3.06)$, then $\displaystyle B_\alpha <F_\alpha <H_\alpha$;
\item if $\alpha_0<\alpha<\alpha_1$, where $\alpha_1\in (3.54,3.55)$, then $ \displaystyle F_\alpha <H_\alpha<B_\alpha $;
\item if $\alpha\geq \alpha_1$, then $\displaystyle H_\alpha <F_\alpha <B_\alpha$.
\end{enumerate}
It turns out that the multiple stable structures for Morse potential, Modified Morse potential and Lennard-Jones potential have been numerically studied in \cite[Fig. 5]{ModifMorse}. Then, we propose the following new conjecture for the Morse potential lattice energy that will be heuristically justified (for the BCC/FCC transition) in Section \ref{sec-conj3d}, based on Sarnak-Str\"ombergsson conjecture \cite[Eq. (43)]{SarStromb} for the lattice theta functions defined by \eqref{def-thetafunction}.

\begin{Conjecture}[Global minimizer of Morse potential lattice energy]\label{ConjMorse3d} Let $r_0>0$, then there exists $\alpha_0$, $\alpha_1$ such that 
\begin{enumerate}
\item If $\alpha<\alpha_0$, then the global minimizer of $E_{\alpha,r_0}$ in $\mathcal{P}_3$ is a BCC lattice.
\item If $\alpha>\alpha_1$, then the global minimizer of $E_{\alpha,r_0}$ in $\mathcal{L}_3$  (resp. in $\mathcal{P}_3$) is a FCC lattice (resp. a HCP structure).
\item If $\alpha_0<\alpha<\alpha_1$, then the unique minimizer of $E_{\alpha,r_0}$ in $\mathcal{P}_3$ is a FCC lattice.
\end{enumerate}
\end{Conjecture}

This conjecture appears to be obviously more difficult to prove than the Sarnak-Str\"ombergsson conjectures for lattice theta functions and Epstein zeta functions described in  \cite[Eq. (43)-(44)]{SarStromb}, in particular for the minimization problem on $\mathcal{P}_3$.

\medskip

We have compared this conjecture to the values of $\alpha$ found from experimental quantities for metals in \cite{GirifalcoWeizer,LincolnetalMorse,SharmaKachhavaMorse,PamuketalMorse,HungetalMorse}. We observe that the real ground-states of metals match with the expected ground-state given by our conjecture, according to the experimental values of $\alpha$, only in few cases for both FCC and BCC structures. Details are given in Section \ref{sec-experiment}, but it clearly appears that the pure central-force model with two-body Morse potential is not sufficiently accurate to describe metals if we want to take into consideration that their FCC or BCC structure are not only local minimizers, but global minimizers. The same holds for rare-gas crystals when we compare the expected ground-state structure with the values of the parameters found in \cite{Raff1990,AlimietalMorse,ParsonMorse,BarkerMorse}.

\medskip

Another interesting fact is the similarity of the two and three-dimensional cases, respectively in the square/cubic cases and the triangular/FCC cases. This justifies the relevance of studying one (two-dimensional) layer instead of the whole crystal, even though the dimension reduction techniques described in \cite{BeterminPetrache} and following from the multiplicative property of the exponential function cannot be applied here. That was also observed in the Lennard-Jones case in \cite{Beterloc,Beterminlocal3d} and we expect this property to be true for many other repulsive-attractive potentials. It is by itself a good motivation to study two-dimensional lattice energies.

\medskip

\textbf{Plan of the paper.} We start in Section \ref{sec-triangoptimal} by showing Theorem \ref{thm-main} and Corollary \ref{thm-triang}, then explaining what are the limits of our method based on Montgomery result and finally proving the local optimality of the triangular lattice at high density. We then prove a generalization of Theorem \ref{prop:stat2d} and Theorem \ref{prop:stat3d} in Section \ref{sec:stat} about Bravais lattices being critical points of energies of type $E_f$ in $\mathcal{L}_d^\circ(V)$, for all $V$ in an open interval. The numerical investigation of the minimizers of $E_{\alpha,r_0}$ is explained in Section \ref{sec-numeric2d}, where the local optimality of the triangular and square lattices is studied and our Conjecture for competitive completely monotone functions is checked. In Section \ref{sec-numeric3d}, the three-dimensional minimization problem is numerically studied and heuristically justified. We also compare the experimental values of $\alpha$ with our Conjecture, explaining why the Morse potential is not a good candidate for describing metals or rare-gas crystals, in the central-force setting.

\section{Optimality of the triangular lattice in $\mathcal{L}_2^\circ(A)$: rigorous results}\label{sec-triangoptimal}

\subsection{Proof of Theorem \ref{thm-main}}

The goal is to study the minimization of $L\mapsto E_{\alpha,r_0}[L]$ in $\mathcal{L}^\circ_2(A)$ for fixed $A$. We want to use the method described in \cite{Betermin:2014fy,BetTheta15} based on Montgomery Theorem \cite[Thm. 1]{Mont} about the optimality of the triangular lattice for the lattice theta function $\theta_L$ defined by \eqref{def-thetafunction}. For that, we need to compute the inverse Laplace transform of $f(r)=e^{\alpha r_0} e^{-2\alpha \sqrt{r}}-2e^{-\alpha \sqrt{r}}$, which is given by
\begin{equation}\label{eq:laplacetransformMorse}
\mu_f(y):=\mathcal{L}^{-1}[f](y)=\frac{\alpha y^{-3/2}}{\sqrt{\pi}}\left( e^{\alpha r_0} e^{-\alpha^2/y}- e^{-\alpha^2/4y}\right).
\end{equation}
Therefore, we now use \cite[Thm 1.1]{BetTheta15} in order to get two sufficient conditions for the optimality of $\Lambda_A$ in $\mathcal{L}_2^\circ(A)$ and, studying the sign of $\mu_f$ and using \cite[Thm 1.5]{OptinonCM}, we get the non-optimality of the triangular lattice in $\mathcal{L}_2^\circ(A)$ for large $A$.

\begin{proof}[Proof of Theorem \ref{thm-main}]
We recall that, in \cite[Thm 1.1]{BetTheta15}, we have proved the following integral representation which we directly apply to $f$, for any $L\in \mathcal{L}_2^\circ(A)$ (note the tiny difference of definition for $\theta_L$ with \cite{BetTheta15} and the fact that the term $p=0$ has to be removed from the sum),
$$
E_{\alpha,r_0}[L]-f(0)=\frac{\pi}{A}\int_1^{+\infty}\left(\theta_L\left( \frac{y}{A} \right)-1 \right)g_A(y)dy + C_A, \quad g_A(y):=y^{-1}\mu_f\left( \frac{\pi}{yA}  \right)+\mu_f\left( \frac{\pi y}{A} \right) ,
$$
where $C_A$ is a constant independent of $L$, $\mu_f$ is the inverse Laplace transform of $f$ and $\theta_L$ is defined by \eqref{def-thetafunction}. In our case, for any $A>0$ and any $y\geq 1$,
\begin{equation}\label{Eq:uA}
g_A(y)=\frac{\alpha A}{\pi^2 y^{3/2}}\left[ e^{\alpha r_0} e^{-\frac{\alpha^2 yA}{\pi}}y^2 - e^{-\frac{\alpha^2 y A}{4\pi}}y^2 + e^{\alpha r_0} e^{-\frac{\alpha^2 A}{\pi y}} -e^{-\frac{\alpha^2 A}{4\pi y}}\right]=:\frac{\alpha A}{\pi^2 y^{3/2}} u_A(y).
\end{equation}
Therefore, if $g_A(y)\geq 0$ for almost every $y\geq 1$, then $\Lambda_A$ is the unique minimizer of $ E_{\alpha,r_0}$ in $\mathcal{L}_2^\circ(A)$. We now show that $(C1)$ and $(C2)$ are both sufficient conditions of the positivity of $g_A$ on $[1,+\infty)$. Let us write $C=e^{\alpha r_0}$ and $\beta=\frac{\alpha^2 A}{\pi}$. We therefore have, for any $y\geq 1$,
\begin{equation}\label{Eq:LBuA}
u_A(y)=Cy^2 e^{-\beta y}-y^2 e^{-\frac{\beta}{4}y} +Ce^{-\frac{\beta}{y}}-e^{-\frac{\beta}{4y}}\geq -y^2 e^{-\frac{\beta}{4}y}+Ce^{-\beta}-1=:h(y).
\end{equation}
We now assume that $Ce^{-\beta}-1>0$, i.e. $A<\frac{\pi r_0}{\alpha}$, otherwise $h$ is clearly negative. We remark that
$$
h(y)\geq 0 \iff g(y):=2\log y -\frac{\beta}{4}y-\log(Ce^{-\beta}-1)\leq 0.
$$
We compute $g'(y)=\frac{2}{y}-\frac{\beta}{4}$, and therefore $g$ is increasing on $[0,8/\beta]$ and decreasing on $[8/\beta,+\infty)$. We thus have two cases:
\begin{enumerate}
\item[(C1)] If $\beta\geq 8$, then $\max_{y\geq 1} g(y)=g(1)=-\frac{\beta}{4}-\log(Ce^{-\beta}-1)$. Therefore, $g(y)\leq 0$ for all $y\geq 1$ if and only if $-\frac{\beta}{4}-\log(Ce^{-\beta}-1)\leq 0$. We then have found that $(C1$) is a sufficient condition for $g_A$ to be positive on $[1,+\infty)$.
\item[(C2)] If $\beta<8$, then $\max_{y\geq 1} g(y)=g(8/\beta)=-2\log \beta -\log(Ce^{-\beta}-1)+2\log 8 -2$. Therefore, $g(y)\leq 0$ for all $y\geq 1$ if and only if $-2\log \beta -\log(Ce^{-\beta}-1)+2\log 8 -2\leq 0$. We then have found that $(C2)$ is a second sufficient condition for $g_A$ to be positive on $[1,+\infty)$.
\end{enumerate}
For the second part of the theorem, it is straightforward to show that $\mu_f(y)\geq 0$ if and only if $y\geq \frac{3\alpha}{4 r_0}$. Therefore, $\mu_f<0$ in a neighbourhood of the origin and, by \cite[Thm 1.5.(1)]{OptinonCM}, $\Lambda_A$ cannot be a minimizer of $E_{\alpha,r_0}$ in $\mathcal{L}_2^\circ(A)$ for sufficiently large $A$.
\end{proof}

\begin{remark}[Numerical values]\label{rem-numtri}
We can now compute the corresponding area bounds for the optimality of the triangular lattice. We fix $r_0=1$ and we compute these bounds for $\alpha \in \{1,2,...,10\}$. in Table \ref{table-C1C2}.

\begin{table}
\begin{center}
\begin{tabular}[H]{|c|c|c|}
\hline
\textbf{$\alpha$} & \textbf{Values of $A$ such that $\Lambda_A$ is minimal for $E_{\alpha,1}$} & \textbf{Condition(s) satisfied} \\
\hline
1 & $\emptyset$ & $\emptyset$\\
\hline
2 & $\emptyset$ & $\emptyset$\\
\hline
3 & $\emptyset$ & $\emptyset$\\
\hline
4 & $[0.1034,0.6782]$ & $(C2)$\\
\hline
5 & $[0.0351,0.5862]$ & $(C2)$\\
\hline
6 & $[0.0139,0.5034]$ & $(C2)$\\
\hline
7 & $[0.0060,0.4378]$ & $(C2)$\\
\hline
8 & $[0.0028,0.3862]$ & $(C2)$\\
\hline
9 & $[0.0013,0.3450]$ & $(C1)$,$(C2)$\\
\hline
10 & $[0.0007,0.3116]$ & $(C1)$,$(C2)$\\
\hline
\end{tabular}
\caption{For $r_0=1$, some values of $A$ such that $\Lambda_A$ is minimal according to Theorem \ref{thm-main}. We actually begin to have a solution for $(C2)$ for (approximatively) $\alpha \geq 3.078$.}\label{table-C1C2}
\end{center}
\end{table}

\end{remark}
\subsection{Proof of Corollary \ref{thm-triang}}
We now give sufficient conditions for $(C1)$ and $(C2)$ to be satisfied and we then prove Corollary \ref{thm-triang}.

\begin{lemma}[Condition for which $(C1)$ is satisfied]\label{lem-C1}
 Let $A_0=-\frac{4\pi}{\alpha^2}\log X_0$ where $X_0$ is the solution of $e^{\alpha r_0}X^4-X-1=0$ on $[e^{-\alpha r_0/4},e^{-2}]$. If $\alpha\geq \frac{8+\log(1+e^{-2})}{r_0}$ and $A$ is such that $\frac{8\pi}{\alpha^2}\leq A\leq A_0<\frac{\pi r_0}{\alpha}$, then $(C1)$ is satisfied.
\end{lemma}
\begin{proof}
Let $X:=e^{-\frac{\alpha^2 A}{4\pi}}\leq 1$. We notice that  $\frac{8\pi}{\alpha^2}\leq A<\frac{\pi r_0}{\alpha}$ if and only if
\begin{equation}\label{condX}
e^{-\frac{\alpha r_0}{4}}<X\leq e^{-2}.
\end{equation}
We now define $P(X):=e^{\alpha r_0}X^4-X-1$ and we want to find $A$ such that \eqref{condX} is satisfied and $P(X)\geq 0$. Since $P'(X)=4e^{\alpha r_0}-1$, then $P$ is decreasing on $[0,4^{-1/3}e^{-\alpha r_0/3})$ and increasing on $(4^{-1/3}e^{-\alpha r_0/3},1]$. Since we assume $\frac{8\pi}{\alpha^2}<\frac{\pi r_0}{\alpha}$, therefore $\alpha>\frac{8}{r_0}$ and we get $\frac{e^{-\frac{\alpha r_0}{3}}}{4^{\frac{1}{3}}}\leq e^{-\frac{\alpha r_0}{4}}<e^{-2}$. It then follows that $P$ is then increasing on $\left[ e^{-\frac{\alpha r_0}{4}},e^{-2}\right]$. We compute $P\left(e^{-\frac{\alpha r_0}{4}} \right)=-e^{-\frac{\alpha r_0}{4}}<0$ and  $P(e^{-2})=e^{\alpha r_0-8}-e^{-2}-1$ which is positive if and only if $\alpha\geq \frac{8+\log(1+e^{-2})}{r_0}$ and the proof is then completed.
\end{proof}
\begin{remark}
It is actually possible to write an exact (complicated) formula for $A_0$ involving $r_0$ and $\alpha$, but we do not need such expression for our purpose here.
\end{remark}

\begin{lemma}[Sufficient conditions for which $(C2)$ is satisfied]\label{lem-C2}
If 
$$
\alpha>\frac{8+\log 2}{r_0} \quad \textnormal{and}\quad \frac{8\pi}{\alpha^2 \sqrt{e^{\alpha r_0-8}-1}}\leq A< \frac{8\pi}{\alpha^2}<\frac{\pi r_0}{\alpha},
$$
 then $(C2)$ is satisfied.
\end{lemma}
\begin{proof}
Let $F(A):= A^2(e^{\alpha r_0}e^{-\frac{\alpha^2 A}{\pi}}-1)-\frac{64\pi^2}{e^2 \alpha^4}$, then, since $A<\frac{8\pi}{\alpha^2}$,
\begin{align*}
F(A)&\geq A^2\left(e^{\alpha r_0}e^{-\frac{\alpha^2}{\pi}\frac{8\pi}{\alpha^2}}-1  \right)-\frac{64\pi^2}{e^2 \alpha^4}=A^2\left(e^{\alpha r_0-8} -1 \right)-\frac{64\pi^2}{e^2 \alpha^4}
\end{align*}
which is nonnegative if $A\geq \frac{8\pi}{\alpha^2 \sqrt{e^{\alpha r_0-8}-1}}$.
\end{proof}

Therefore, combining Lemmas \ref{lem-C1} and \ref{lem-C2}, Theorem \ref{thm-triang} is proved.

\begin{remark}\label{rmk:globaltriimp}
For any fixed $\alpha$, as $r_0\to +\infty$, $A_0\to \frac{\pi r_0}{\alpha}<r_0^2$. Therefore, it is impossible to prove
 that $\Lambda_A$ is the unique minimizer of $E_{\alpha,r_0}$ in $\mathcal{L}_2^\circ(A)$ for any $A\in (0,r_0^2]$. We notice that it is straightforward (see e.g. \cite[Step 3 p. 3252]{BetTheta15}) to show that the global minimizer of $E_{\alpha,r_0}$ in $\mathcal{L}_2$ must have an area smaller than $r_0^2$. Therefore, it is not possible to conclude, whatever $r_0$ is, that the global minimizer of $E_{\alpha,r_0}$ in $\mathcal{L}_2$ is a triangular lattice.
\end{remark}

\subsection{Limits of our method}\label{sec-limit}

As we have already explained in \cite[Sec. 4.3]{BetTheta15}, our method based on Montgomery Theorem \cite[Thm 1]{Mont} is not optimal. Even though we were quite successful with it for Lennard-Jones type potentials and some difference of Yukawa potentials in \cite[Thm 1.2]{BetTheta15}, it turns out that, contrary to these examples:
\begin{enumerate}
\item We cannot prove the minimality of the triangular lattice at high density for $E_{\alpha,r_0}$ (i.e. for arbitrarily small $A$).
\item We cannot conclude that the global minimizer of $E_{\alpha,r_0}$ in $\mathcal{L}_2$ is a triangular lattice, even for a single value of $(\alpha,r_0)$.
\end{enumerate}
In this section, we show why the high density minimality cannot be reached with our method and we propose a modification of the Morse potential for getting this optimality by using this method.

\medskip

The next results shows that for all $\alpha>0$, we can find a large value $y_0\geq 1$ such that $g_A(y_0)$ is negative for small enough $A$. We recall that $g_A(y)=\frac{\alpha A}{\pi^2 y^{3/2}}u_A(y)$ (see \eqref{Eq:uA}).

\begin{lemma}[Negativity of $u_A$ for small $A$]
For any $\alpha>0$ and $r_0>0$, there exists $\lambda$ such that
$$
\lim_{A\to 0} u_A\left(  \frac{\lambda}{A}\right)=-\infty.
$$
\end{lemma}
\begin{proof}
We easily compute that
$$
\lim_{A\to 0} u_A\left(  \frac{\lambda}{A}\right)=\lim_{A\to 0} \frac{\lambda^2}{A^2}\left( e^{\alpha r_0} e^{-\frac{\alpha^2 \lambda}{\pi}}-e^{-\frac{\alpha^2 \lambda}{4\pi}}\right) +e^{\alpha r_0}-1.
$$
Furthermore, we have that 
$$
e^{\alpha r_0} e^{-\frac{\alpha^2 \lambda}{\pi}}-e^{-\frac{\alpha^2 \lambda}{4\pi}}<0 \iff \alpha>\frac{4\pi r_0}{3\lambda}.
$$
Therefore, for any $\alpha$ and $r_0$, there exist $\lambda$ such that $\alpha r_0>4\pi/3\lambda$, thus the result is proved.
\end{proof}


It is also possible to modify the Morse potential by adding an inverse power law to its expression in order to get the optimality of the triangular lattice at high density. 

\begin{prop}[Modification of $f$ and optimality of $\Lambda_A$ for small $A$] Let $p>5/2$ and $\beta>0$, we then define the following modification of the Morse potential:
$$
\tilde{f}(r^2)=f(r^2)+\frac{\beta}{r^{2p}}=e^{\alpha} e^{-2\alpha r}-2e^{-\alpha r}+\frac{\beta}{r^{2p}}. 
$$
For any $\beta>0$, $\alpha>0$ and $p>5/2$, if $0<A\leq \min\left\{\frac{\pi}{\alpha}, \left( \frac{\beta \pi^{p+1}}{\alpha \Gamma(p)} \right)^{\frac{1}{p}}\right\}$, then $\Lambda_A$ is the unique minimizer in $\mathcal{L}_2^\circ(A)$, up to rotation, of 
$$
L\mapsto E_{\tilde{f}}[L]=\sum_{p\in L\backslash \{0\}} \tilde{f}(|p|^2).
$$
\end{prop}
\begin{remark}
The potential $r\mapsto \tilde{f}(r^2)$ can be also seen as a small (exponential) perturbation of the opposite of the Buckingham type potential $V_B(r)=2e^{-\alpha r} -\beta r^{-2p}$ (see \cite[Sec. 7.2]{BetTheta15}) originally proposed by Buckingham for modelling interactions in rare-gases (see \cite[p. 276]{Buck}).
\end{remark}
\begin{proof}
It is again a straightforward application of \cite[Thm 1.1]{BetTheta15}, using the estimate we found in the proof of Theorem \ref{thm-triang}. More precisely, we have, for any Bravais lattices $L$ of area $A$ (note again the tiny difference of definition for $\theta_L$ with \cite{BetTheta15}),
$$
E_{\tilde{f}}[L]=\frac{\pi}{A}\int_1^{+\infty}\left(\theta_L\left( \frac{y}{A} \right)-1 \right) \tilde{g}_A(y)dy + C_A, \quad \tilde{g}_A(y):=y^{-1}\mu_{\tilde{f}}\left( \frac{\pi}{yA}  \right)+\mu_{\tilde{f}}\left( \frac{\pi y}{A} \right).
$$
It is easy to compute $\tilde{g}_A(y)=\frac{\alpha A}{\pi^2 y^{p-1/2}}\left(u_A(y) y^{p-5/2} +\frac{\beta \pi^{p+1}}{\alpha \Gamma(p)}A^{-p}+\frac{\beta \pi^{p+1}}{\alpha \Gamma(p)}A^{-p}y^{2p-2}  \right)$ and to show that, assuming that $A<\frac{\pi}{\alpha}$,
\begin{align*}
\tilde{g}_A(y) &\geq \frac{\alpha A}{\pi^2 y^{p-1/2}}\left( -y^{p-1/2}+\left(e^{\alpha}e^{-\frac{\alpha^2 A}{\pi}}-1 \right)y^{p-5/2}+   \frac{\beta \pi^{p+1}}{\alpha \Gamma(p)}A^{-p}y^{2p-2} + \frac{\beta \pi^{p+1}}{\alpha \Gamma(p)}A^{-p}\right)\\
&=:\frac{\alpha A}{\pi^2 y^{p-1/2}}P_A(y).
\end{align*}
Using Cauchy's upper bound for the largest root of $P_A$ (see \cite[Sec. 2.4]{BetTheta15}), we find that $y\geq \left( \frac{\Gamma(p) \alpha A^p}{\beta \pi^{p+1}} \right)^{\frac{1}{p-3/2}}\Rightarrow P_A(y)\geq 0\Rightarrow \tilde{g}_A(y)\geq 0$. Since we have $\left( \frac{\Gamma(p) \alpha A^p}{\beta \pi^{p+1}} \right)^{\frac{1}{p-3/2}}\leq 1$ by assumption, $\tilde{g}_A(y)\geq 0$ if $y\geq 1$ and the proof is completed.
\end{proof}

\begin{example}
 By choosing $p$ very large and $\beta$ very small, we get a reasonable approximation of $V_M$ (close and after its minimum) for which we can prove the optimality of the triangular lattice at high density. As an example, we have plotted $V_M$ for $r_0=1$ and $\alpha=6$ as well as $r\mapsto V_M(r)+r^{-12}=e^6\tilde{f}(r)$ with $p=100$ and $\beta=(10 e)^{-6}$ in Figure \ref{fig:modifmorse}. Applying the previous proposition, we get the optimality of $\Lambda_A$ for $E_{\tilde{f}}$ in $\mathcal{L}_2^\circ(A)$ for any $0<A<0.07056$ and Theorem \ref{thm-main} gives the optimality of the triangular lattice for $E_{6,1}$ when $0.0139\leq A\leq 0.5034$.
 \end{example}

\begin{figure}[!h]
\centering
\includegraphics[width=7cm]{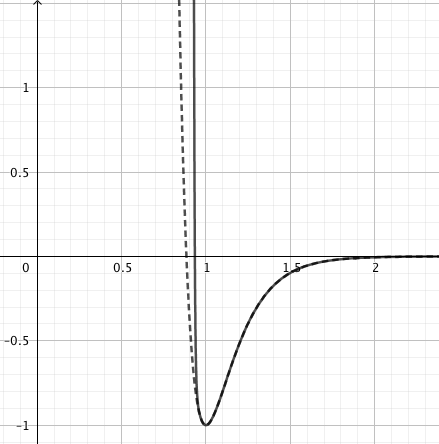} 
\caption{Plot of $V_M$ for $\alpha=6,r_0=1$ (dashed line) and $r\mapsto e^6\tilde{f}(r^2)$ with $\beta=(10e)^{-6}$, $p=100$.}
\label{fig:modifmorse}
\end{figure}

\begin{remark} Our method seems to be effective for proving the optimality of $\Lambda_A$ when $A$ is arbitrarily small when the interaction potential diverges and is equivalent to a completely monotone function (like an inverse power law) at $0$, as we have already remarked in \cite{Betermin:2014fy,BetTheta15}. We also notice that $\tilde{f}$ is not a function with only one well, hence it is again impossible to apply our method developed in \cite{Betermin:2014fy,BetTheta15} to show that the global minimizer of $E_{\tilde{f}}$ is a triangular lattice (see Remark \ref{rmk:globaltriimp}).
\end{remark}

\subsection{Local minimality of $\Lambda_A$ for small $A$}

Using lattice symmetries, it is straightforward to show that, for any $\alpha,r_0,A\in (0,+\infty)$, $\Lambda_A$ and $\sqrt{A}\Z^2$ are critical points of $E_{\alpha,r_0}$ in $\mathcal{L}_2^\circ(A)$ (see e.g. \cite[Prop 3.2 and 3.4]{Beterloc}). We also recall the following result showing that the Hessian of our energy at $\Lambda_A$ has a very simple diagonal form.

\begin{lemma}[\cite{Beterloc}]\label{derivtriangular}
The Hessian of $E_{\alpha,r_0}$ at $\Lambda_A$ is $D^2 E_{\alpha,r_0}[\Lambda_A]=T_{\alpha,r_0}(A) I_2$ where
\begin{align*}
T_{\alpha,r_0}(A):=\frac{4A}{\sqrt{3}}\sum_{m,n}n^2 f'\left(\frac{2A}{\sqrt{3}}[m^2+mn+n^2]  \right)+\frac{4A^2}{3}\sum_{m,n}n^4f''\left(\frac{2A}{\sqrt{3}}[m^2+mn+n^2]  \right).
\end{align*}
\end{lemma}
Therefore, substituting $f$ by its expression, writing $Q(m,n):=m^2+mn+n^2$ and $\ell_A:=\sqrt{\frac{2A}{\sqrt{3}}}$, we obtain
\begin{align*}
T_{\alpha,r_0}(A)=&A\alpha^2 \sum_{m,n} \frac{n^4 e^{-\alpha \ell_A\sqrt{Q(m,n)}}}{Q(m,n)}\left\{\frac{2}{\sqrt{3}}e^{\alpha r_0}e^{-\alpha \ell_A \sqrt{Q(m,n)}}-1  \right\}\\
&\quad +\frac{\sqrt{A}\alpha}{\sqrt{2}3^{1/4}}\sum_{m,n} \frac{n^2 e^{-\alpha \ell_A\sqrt{Q(m,n)}}}{\sqrt{Q(m,n)}}\left(  4-\frac{n^2}{Q(m,n)}\right)\left\{1-e^{\alpha r_0} e^{-\alpha \ell_A \sqrt{Q(m,n)}}  \right\}
\end{align*}

Thus, it is possible to show that $\Lambda_A$ is a local minimizer of $E_{\alpha,r_0}$ in $\mathcal{L}_2^\circ(A)$ for sufficiently small values of $A$. This is a good indication about the optimality of the triangular lattice at high density, which is impossible to get by using our method from \cite{Betermin:2014fy,BetTheta15} as recalled in the previous section.

\begin{lemma}[Local optimality of the triangular lattice at high density]\label{lem:asympttriangular}
There exists $A_0$ such that for any $0<A<A_0$, $\Lambda_A$ is a local minimizer of $E_{\alpha,r_0}$ in $\mathcal{L}_2^\circ(A)$.
\end{lemma}
\begin{proof}
We first recall that $\Lambda_A$ is a critical point in $\mathcal{L}_2^\circ(A)$. Furthermore, we can write $T_{\alpha,r_0}(A)=\sqrt{A}\alpha\left(E_1(A,\alpha)+ \sqrt{A}\alpha E_2(\alpha,A)  \right)$, thus the sign of $T_{\alpha,r_0}(A)$ as $A\to 0$ is given by 
$$
E_1(\alpha,A):=\sum_{m,n} \frac{n^4 e^{-\alpha \ell_A\sqrt{Q(m,n)}}}{Q(m,n)}\left\{\frac{2}{\sqrt{3}}e^{\alpha r_0}e^{-\alpha \ell_A \sqrt{Q(m,n)}}-1  \right\}>0
$$
for $A$ small enough, because $\frac{2}{\sqrt{3}}e^{\alpha r_0}-1>0 $ and $\frac{n^4 e^{-\alpha \ell_A\sqrt{Q(m,n)}}}{Q(m,n)}$ is decreasing exponentially fast as $m^2+mn+n^2$ increases.
\end{proof}

It turns out that a systematic analysis of the sign of $T_{\alpha,r_0}(A)$, as well as all the Hessians in dimensions $2$ and $3$, with respect to the area $A$ (or the volume) is a difficult task. Therefore, we have performed in Section \ref{sec-numeric2d} and \ref{sec-numeric3d} many numerical investigations showing the nature of the main Bravais lattices we are interested in.

\section{A general result about volume-stationary lattices - Proof of Theorems \ref{prop:stat2d} and \ref{prop:stat3d}}\label{sec:stat}

In this section, we show a general result about `volume-stationary lattices', i.e. lattices being critical points of $E_f$ defined by
$$
E_f[L]:=\sum_{p\in L\backslash \{0\}} f(|p|^2)
$$
on $\mathcal{L}_d^\circ(V)$ in an open interval of volumes $V$, for a large class of interaction potential $f:(0,+\infty)\to \R$ integrable at infinity and being the Laplace transform of a Borel measure on $(0,+\infty)$. It is important to notice that all the classical interaction potentials used in molecular simulations belong to this class of functions. Our result is based on Gruber's results \cite[Cor 1 and Cor 2]{Gruber} where all the lattices being critical points for the Epstein zeta function, defined by
$$
\zeta_L(s):=\sum_{p\in L \backslash \{0\}} \frac{1}{|p|^s}
$$
and for all $s>d$, are characterized. It turns out that they all have their layers strongly eutactic in the sense of the following definition.

\begin{defi}[Strongly eutactic layer]\label{def:eutactic}
Let $L\in \mathcal{L}_d^\circ(1)$. We say that a layer $M=\{p\in L ; |p|=\lambda\}$, for some $\lambda>0$, of $L$ is strongly eutactic if $\sharp M=2k$ and, for any $x\in \R^d$, 
$$
\sum_{p\in M} \frac{(p\cdot x)^2}{|p|^2}=\frac{2k}{d}|x|^2.
$$
\end{defi}
\begin{remark}
After a suitable renormalization, $M$ is also called a spherical $2$-design (see \cite{BachocVenkov}).
\end{remark}

We then show the following result describing the only Bravais lattices in dimensions $2$ and $3$ that can stay stationary for $E_f$ under any small perturbation of the density. This result confirms our numerical observations performed in this paper as well as in \cite{Beterloc,Beterminlocal3d} for the classical Lennard-Jones potential.

\begin{thm}[Volume-stationary lattices for $E_f$]\label{thm:eutactic}
Let $d\geq 2$ and $f:(0,+\infty)\to \R$ be a nonzero function such that
\begin{enumerate}
\item as $r\to +\infty$, we have $f(r)=O(r^{-d/2-\eta})$ for some $\eta>0$,
\item for any $r>0$, it holds $\displaystyle f(r)=\int_0^{+\infty} e^{-rt} d\mu_f(t)dt$ for some Borel measure $\mu_f$ on $(0,+\infty)$.
\end{enumerate}
Let $L_0\in \mathcal{L}_d^\circ(1)$. There exists an open interval $I$ such that $\nabla_{L_0} E_f[V^{1/d} L_0]=0$ for all $V\in I$ if and only if all the layers of $L_0$ are strongly eutactic, and it follows that $I=(0,+\infty)$.

\medskip

In particular, $L_0\in \{\Z^2,\Lambda_1\}$ in dimension $2$ and $L_0\in \{\Z^3,\mathsf{D}_3,\mathsf{D}_3^*\}$ in dimension $3$.
\end{thm}
\begin{remark}
As we will see in the proof, it also means that $L_0$ is a critical point of $\theta_L(\alpha)$ for almost all $\alpha>0$, where the lattice theta function is defined by \eqref{def-thetafunction}. Furthermore, in higher dimension, as explained in \cite[Cor 2]{Gruber}, there are only finitely many such lattices. 
\end{remark}
\begin{proof}
Let $L_0\in \mathcal{L}_d^\circ$. We assume that $\nabla_{L_0} E_f[V^{1/d}L_0]=0$ for any $V\in I$ where $I\subset \R$ is an open interval. We easily show, using Fubini's theorem and the definition of $f$, that, for any $L\in \mathcal{L}_d^\circ(1)$, and where $\theta_L$ is defined by \eqref{def-thetafunction},
\begin{align*}
E_f[V^{1/d} L]&=\sum_{p\in L\backslash \{0\}} f\left(V^{2/d} |p|^2\right)=\int_0^{+\infty} \left[ \theta_L\left( \frac{V^{2/d}t}{\pi} \right)-1  \right]d\mu_f(t)
\end{align*}
We therefore get, by Lebesgue's dominated convergence Theorem,
$$
\nabla_{L_0} E_f[V^{1/d}L_0]=\int_0^{+\infty} \nabla_{L_0} \theta_{L_0}\left(\frac{V^{2/d}t}{\pi}  \right)d\mu_f(t).
$$
By analyticity of $\alpha\mapsto \theta_L(\alpha)$ for any $L$, it follows that all the components of $V^{2/d}\mapsto \nabla_{L_0} E_f[V^{1/d}L_0]$ are also analytic functions on $(0,+\infty)$. Thus, for each component, its set of zeros is a discrete set, which implies that $\nabla_{L_0} E_f[V^{1/d}L_0]=0$ for all $V>0$, proving that $I=(0,+\infty)$. It follows that $\nabla_{L_0} \theta_{L_0}\left(\frac{t}{\pi}  \right)=0$ for almost every $t>0$, i.e. $L_0$ is a critical point of $L\mapsto \theta_L(\alpha)$ for almost every $\alpha>0$. Since $\zeta_L(s)=E_{f_s}[L]$ for $f_s(r)=\frac{1}{r^{s/2}}$ that belongs to the space of functions defined in the statement of our theorem, it follows that $L_0$ is a critical point of the Epstein zeta function $L\mapsto \zeta_L(s)$ for all $s>d$. By \cite[Cor 1]{Gruber}, the only such lattices have their layers being strongly eutactic and the result is proved because all these lattices are critical points for $E_f$ for all fixed volume as proved in \cite[Thm 4.4]{CoulSchurm2018}. Furthermore, the strongly eutactic lattices in dimensions $2$ and $3$ are recalled in \cite[Cor 2]{Gruber} and are the square, triangular, simple cubic, FCC and BCC lattices.
\end{proof}

\section{Numerical investigations in dimension 2}\label{sec-numeric2d}
In this part, as in the next one treating the three-dimensional minimization problem, we want to understand how the minimizer of $E_{\alpha,r_0}$ in $\mathcal{L}_2^\circ(A)$ changes when $A$ varies as well as the nature of its global minimizer in $\mathcal{L}_2$. We also want to compare our numerical observations with the one we have performed in \cite{Beterloc} for the classical Lennard-Jones potential $V_{LJ}(r)=\frac{1}{r^{12}}-\frac{2}{r^6}$. We then have to choose properly the values of $\alpha$ and $r_0$. We have decided that the parameters have to satisfy $V_{LJ}(1)=f(1)=\min_r f(r)=-1$ and $V_{LJ}''(1)=f''(1)$, i.e. we want $r_0=1$ and $\alpha=6$. Figure \ref{fig-LJMorse} shows the graph of both potentials. They obviously are very similar in a small neighbourhood of $r=1$, but $f$ is decreasing faster to $0$ as $r$ increases and is equal to $e^{6}-2\approx 401.4$ for $r=0$.\\
\begin{figure}[!h]
\centering
\includegraphics[width=7cm]{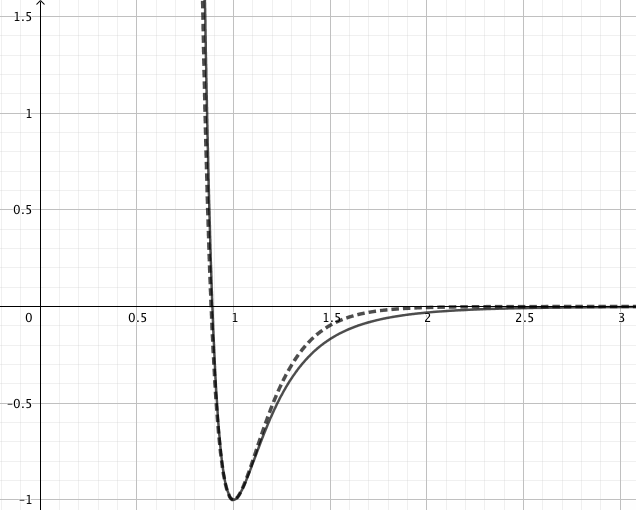} 
\caption{Plot of $V_M$ for $\alpha=6,r_0=1$ (dashed line) and $V_{LJ}(r)=r^{-12}-2r^{-6}$.}
\label{fig-LJMorse}
\end{figure}

As we will see in this section, we also have performed the same investigation for many values of $\alpha,r_0$, and in particular for $\alpha=3$ and $r_0=3$ (see Section \ref{subsec-Conj}). It turns out that the transition between the minimizers of $E_{\alpha,r_0}$ with respect to the area $A$ is more clear in the latter case.

\subsection{Local optimality of the triangular and square lattices}

The numerical study of the Hessian's sign can be easily done for the triangular lattice $\Lambda_A$ and the square lattice $\sqrt{A}\Z^2$. We already know that these lattices are critical points for $E_{\alpha,r_0}$ for all $A>0$ and that the Hessian is a multiple of the identity in the triangular lattice case (see Lemma \ref{derivtriangular}). For the square lattice, it turns out that, again for symmetry reasons (see \cite[Cor. 3.8]{Beterloc}), the Hessian is also diagonal, thus we only have again to study the sign of its eigenvalues. Our results are summarized in Table \ref{table:localtrisquare}.

\medskip

\begin{table}
\begin{center}
\begin{tabular}[H]{|c|c|c|c|c|}
\hline
 \textbf{Lattice} & \textbf{$\Lambda_A$ (M)} & \textbf{$\Lambda_A$ (LJ)} & \textbf{$\sqrt{A}\Z^2$ (M)} & \textbf{$\sqrt{A}\Z^2$ (LJ)}\\
 \hline
 \textbf{Local minimum} & $A<1.175$ & $A<1.152$ & $1.555<A<1.285$ & $1.143<A<1.268$\\
 \hline
 \textbf{Local maximum} & $A>1.175$ & $A>1.152$ & $\emptyset$ & $\emptyset$\\
 \hline
 \textbf{Saddle point} & $\emptyset$ & $\emptyset$ & $A\not\in (1.555,1.1285)$ & $A\not \in (1.143,1.268)$\\
 \hline
\end{tabular}
\caption{Local optimality of $\Lambda_A$ and $\sqrt{A}\Z^2$ for Morse potential (M) with $r_0=1$, $\alpha=6$, and Lennard-Jones (LJ) potential for which the values are taken from \cite{Beterloc}}\label{table:localtrisquare}
\end{center}
\end{table}

\subsection{Confirmation of our conjecture in dimension 2}\label{subsec-Conj}

As in \cite{Beterloc}, we have investigated the minimizer of $E_{\alpha,r_0}$ in $\mathcal{L}_2^\circ(A)$ when the area $A$ varies. Our study is again based on the investigation of this minimizer among rhombic and rectangular lattices defined by \eqref{rectangular}-\eqref{rhombic}. Our results are summarized in Table \ref{table-conj} and again compared to the one found for the classical Lennard-Jones potential in \cite{Beterloc}. As in the Lennard-Jones case, we observe that there is a transition triangular-rhombic-square-rectangular as $A$ increases. Certainly due to the exponential decay of the Morse potential, all the transitions appear earlier than in the Lennard-Jones case. Furthermore, we also observe a transition from a triangular lattice to a rhombic lattice with an angle slightly larger than $60^\circ$ around $A\approx 1.1560011044$. Then we have observed that the minimizer becomes extremely quickly and continuously a square lattice, as $A$ increases, for $A\approx 1.1560011045$. This transition with a discontinuous jump is better observed in the case $\alpha=3$, $r_0=3$, where it appears for $A\approx 9.285$, the transition being from a triangular lattice to a rhombic lattice with an angle $\theta\approx 72.19^\circ$ (see Figure \ref{fig-rhombic}), and continuously becoming a square lattice at $A\approx 9.4$. Moreover, as proved in Theorem \ref{prop:stat2d}, we also observe that the triangular and square lattices are the only one being minimizers in an open interval of areas and, moreover, that the minimizer of $E_{\alpha,r_0}$ in $\mathcal{L}_2^\circ(A)$ becomes a rectangular lattice arbitrarily thin as $A$ goes to $+\infty$ (see Figure \ref{fig-rect}). Thus, for the Morse potential, the global behaviour of the minimizer of $E_{\alpha,r_0}$ as $A$ varies seems to be qualitatively exactly the same as for the Lennard-Jones potential, confirming our Conjecture \cite[Sec. 5.4]{Beterloc} already recalled in the introduction of this paper.

\medskip

Furthermore, we observe that, for any $A\leq 1$, the triangular lattice seems to be the minimizer of $E_{6,1}$ in $\mathcal{L}_2^\circ(A)$. Since the minimum of $f$ is achieved for $r=r_0=1$, it is easy to prove (see e.g. \cite[Step 3 p. 3252]{BetTheta15}) that the area $A_0$ of the global minimizer of $E_{6,1}$ in $\mathcal{L}_2$ satisfies $A_0\leq 1$. Therefore, it seems to be numerically clear that this global minimizer is a triangular lattice, as we expect for the classical Lennard-Jones energy and as we have already proved for some Lennard-Jones-type potentials with small parameters in \cite[Thm 1.2]{BetTheta15}. This results is conjectured in \cite[Sec. 5.4]{Beterloc} to be true for any difference $f=g-h$ of completely monotone potential such that $f$ has one well.

\begin{figure}
\centering
\includegraphics[width=8cm]{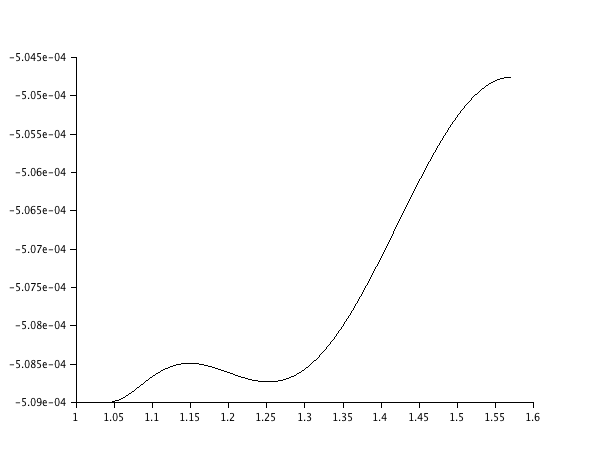} \quad \includegraphics[width=8cm]{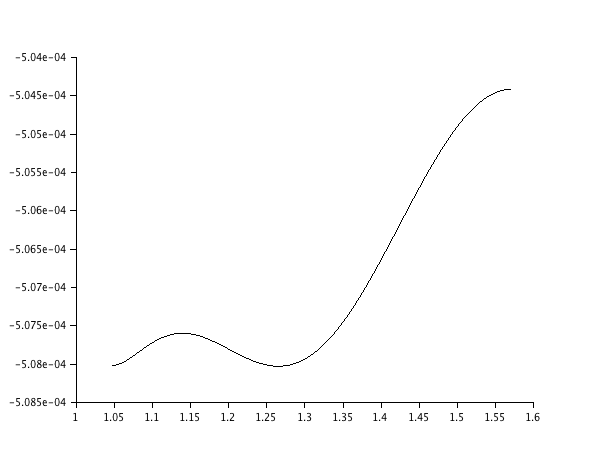}
\caption{For $\alpha=3$, $r_0=3$, plot of $\theta\mapsto E_{3,3}[L_\theta]$, where $L_\theta$ is a rhombic lattice with angle $\theta$ (see \eqref{rhombic}), on $[\pi/3,\pi/2]$. We observe a transition of the minimizer from a triangular lattice when $A=9.28$ (left) to a rhombic lattice with an angle $\theta \approx 72.19^\circ$ when $A=9.285$ (right).}
\label{fig-rhombic}
\end{figure}

\begin{figure}
\centering
\includegraphics[width=8cm]{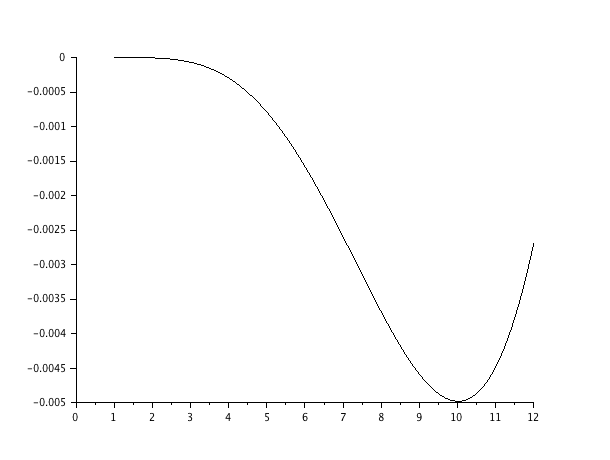} \quad \includegraphics[width=8cm]{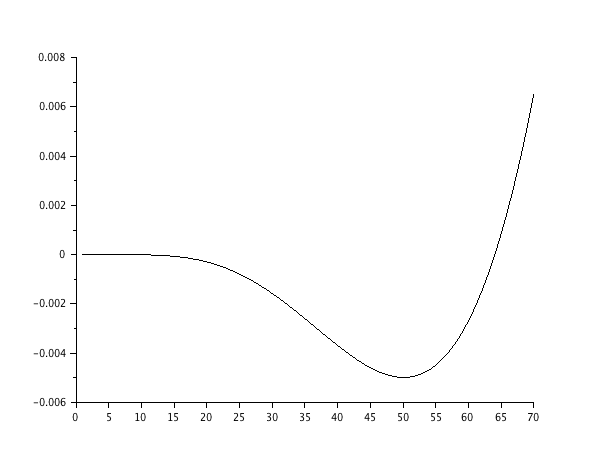}
\caption{The minimizer of $E_{6,1}$ becomes more and more thin as $A\to +\infty$. We have plotted $y\mapsto E_{6,1}[\sqrt{A}L_y]$ where $L_y$, $y\geq 1$ is a rectangular lattice (see \eqref{rectangular}), for $A=10$ (left) and $A=50$ (right). We have computed the minimum in many other cases and it seems that the minimum $y_A$ of $y\mapsto E_{6,1}[\sqrt{A}L_y]$ satisfied $y_A\approx A$, showing that the minimizer, for large $A$, is becoming more and more thin as $A$ increases.}
\label{fig-rect}
\end{figure}

\begin{remark}[The large $\alpha$ case in any dimension]
We observe that the Morse potential $V_M$, for $r_0=1$ fixed, converges to $-\delta(x-1)$ as $\alpha\to +\infty$ (see Figure \ref{fig:Morseplots}). A similar proof as the one we have done for the Lennard-Jones type potential in \cite[Thm 1.13]{OptinonCM} shows the global optimality of a lattice achieving the kissing number with balls of radius $1/2$, for sufficiently large $\alpha$. In particular, in dimension $2$ (resp. 3), the global minimizer of $E_{\alpha,r_0}$ in $\mathcal{L}_d$ for large $\alpha$ is a triangular (resp. FCC) lattice.
\end{remark}

\section{Numerical investigations in dimension 3}\label{sec-numeric3d}
We now give the results of our investigation for the minimization of $E_{\alpha,r_0}$ in $\mathcal{L}_3$ and some comparisons of the Morse energy for BCC/FCC lattices and HCP structure.

\subsection{Local optimality of the cubic lattices}
As in dimension 2, for symmetry reasons, it is straightforward to prove that all the cubic lattices are critical point of the Morse energy, i.e. for any $\alpha,r_0,V\in (0,+\infty)$, $V^{\frac{1}{3}}\Z^3$, $V^{\frac{1}{3}}\mathsf{D}_3$ and $V^{\frac{1}{3}}\mathsf{D}_3^*$ are critical points of $E_{\alpha,r_0}$ in $\mathcal{L}_3^\circ(V)$ (see \cite[Prop 3.2]{Beterminlocal3d}).

\medskip

We have first investigated the local optimality of these cubic lattices, as in \cite{Beterminlocal3d}, for $E_{\alpha,r_0}$ with respect to the volume $V$ of their unit cells. We have used the formulas proved in \cite[Sec. 4]{Beterminlocal3d}. We can see that the results are quite similar with what we get in the classical Lennard-Jones case, except for the BCC lattice $V^{\frac{1}{3}}\mathsf{D}_3^*$ for which there is a difference with the FCC lattice. This is basically due to the lack of homogeneity of the Morse potential and also to the minimizer transition for the exponential sum $\theta_{\sqrt{L}}$ described in Section \ref{sec-conj3d}. We also notice that this cubic lattices are the only one being critical points of $E_{\alpha,r_0}$ in $\mathcal{L}_3^\circ(V)$ for $V$ being in some open intervals of volumes (see Theorem \ref{prop:stat3d}).

\begin{table}
\begin{center}
\begin{tabular}[H]{|c|c|c|c|}
\hline
 \textbf{Lattice} & \textbf{$V^{\frac{1}{3}}\mathsf{D}_3$}  & \textbf{$V^{\frac{1}{3}}\mathsf{D}_3^*$}  & \textbf{$V^{\frac{1}{3}}\Z^3$} \\
 \hline
 \textbf{Local min} & (M): $V<1.125$ & (M): $V<0.33$ & (M): $1.215<V<1.425$\\
 & (LJ): $V<1.091$ & (LJ): $V< 1.091$& (LJ):$1.2<V<1.344$  \\
 \hline
  \textbf{Local max} &(M): $V>1.375$ & (M): $V>1.215$ & (M): $\emptyset$\\
 & (LJ): $V>1.313$ & (LJ): $V>1.313$ &(LJ): $\emptyset$\\
  \hline
   \textbf{Saddle point}  & (M): $1.125<V<1.375$ & (M): $0.33<V<1.215$ & (M): $V\not \in (1.215,1.425)$\\
 & (LJ): $1.091<V<1.313$ & (LJ): $1.091<V<1.313$ & (LJ): $V\not\in (1.2,1.344)$ \\
   \hline
\end{tabular}
\caption{Values of $V$ such that the cubic lattices $V^{\frac{1}{3}}\mathsf{D}_3, V^{\frac{1}{3}}\mathsf{D}_3^*$ and $V^{\frac{1}{3}}\Z^3$ are local optimizers for $E_{6,1}$ (notation: (M)) and $E_{V_{LJ}}$ (notation: (LJ)) for which the values are taken from \cite{Beterminlocal3d}.}\label{table-3dcubic}
\end{center}
\end{table}
\begin{remark}\label{rem-alpha3locmin}
We also notice that, for $\alpha=3$ and $r_0=1$, $V^{\frac{1}{3}}\mathsf{D}_3^*$ is a local minimum of $E_{3,1}$ in $\mathcal{L}_3^\circ(V)$ if and only if $V\in (0,V_0]$ where $V_0\approx 1.085$. 
\end{remark}

\subsection{Comparison of possible global minimizers}
We first give a result explaining the behaviour of $E_{\alpha,r_0}$ among the dilated version of a given lattice $L$.

\begin{lemma}\label{lem-dilationenergy}
For any fixed lattice $L\in \mathcal{L}_3$ and any $\alpha,r_0\in (0,+\infty)$, the function $f_L(\lambda):=E_{\alpha,r_0}[\lambda L]$ is decreasing on $(0,\lambda_0)$ and increasing on $(\lambda_0,+\infty)$ for some $\lambda_0$ depending on $\alpha,r_0,L$.
\end{lemma}
\begin{proof}
We have
$$
f_L'(\lambda)=2\alpha\left( \sum_{p\in L} |p| e^{-\alpha \lambda |p|} -e^{\alpha r_0}\sum_{p\in L} |p| e^{-2\alpha \lambda |p|}\right)
$$
and 
$$
f_L'(\lambda)=0 \iff e^{\alpha r_0}=\frac{g_L(\lambda)}{g_L(2\lambda)},\quad g_L(\lambda):=\sum_{p\in L} |p| e^{-\alpha \lambda |p|}> 0.
$$
It is now clear, by comparison of exponential growth, that, for any fixed $\alpha,r_0>0$, $\lambda\mapsto  \frac{g_L(\lambda)}{g_L(2\lambda)}$ is strictly increasing and has its values on $(0,+\infty)$, both bounds corresponding to the limit of $g_L$ as $\lambda\to 0$ and $\lambda\to +\infty$. Therefore, there exists a unique $\lambda_0$ such that $f_L'(\lambda_0)=0$ and it is easy to see that $f_L'(\lambda)<0$ (resp. $f_L'(\lambda)>0$) as $\lambda\to 0$ (resp. $\lambda\to +\infty$). It follows that $f_L$ is decreasing on $(0,\lambda_0)$ and increasing on $(\lambda_0,+\infty)$.
\end{proof}

Therefore, comparing the energies of $\lambda L$ among all the $\lambda>0$ and for different lattices $L$, we are able to numerically find what is the good candidate for the global minimization of $E_{\alpha,r_0}$ among these structures. Fixing $r_0=1$, we have investigated the energy of $\lambda \mathsf{D}_3$, $\lambda \mathsf{D}_3^*$ and $\lambda \hcp$ for $\lambda>0$ and different values of $\alpha$. We have defined $H_\alpha:=\min_\lambda E_{\alpha,1}[\lambda \hcp]$, $B_\alpha:=\min_\lambda E_{\alpha,1}[\lambda \mathsf{D}_3^*]$ and $F_\alpha:=\min_\lambda E_{\alpha,1}[\lambda \mathsf{D}_3]$. Using Lemma \ref{lem-dilationenergy} to be sure to localize all minima, we then have observed the following:
\begin{enumerate}
\item For any $\alpha\in \{3+0.01k : k\in \{0,1,...,5\}\}$, then $B_\alpha<F_\alpha<H_\alpha$,
\item For any $\alpha \in \{ 3+0.01k : k\in \{6,...,40\}\}$, then $F_\alpha<H_\alpha<B_\alpha$,
\item For any $\alpha\in \{3.5,4,5,6,7,8,9,10\}$, then $H_\alpha<F_\alpha<B_\alpha$.
\end{enumerate}
These results support our Conjecture \ref{ConjMorse3d} and the BCC/FCC phase transition is heuristically explained in the next section.

\begin{remark}[Local minimality of the probable global BCC minimizer]
It is also important to notice how close are the values of the minimal energies. For example, in the $\alpha=3$ case, we have $|B_3- F_3|<5\times 10^{-4}$. Furthermore, according to Remark \ref{rem-alpha3locmin} and the fact that a global minimizer of any $E_{\alpha,1}$ must have a volume smaller than $1$, we know that $V_m^{1/3}\mathsf{D}_3^*$ -- where $V_m$ is the volume minimizing $V\mapsto E_{3,1}[V^{1/3}\mathsf{D}_3^*]$ -- is a local minimum of $E_{3,1}$ on $\mathcal{L}_d$, i.e. the expected global minimizer of $E_{3,1}$ is a local minimizer. The same can be shown in all the previously stated cases.
\end{remark}

\subsection{Heuristic arguments supporting Conjecture \ref{ConjMorse3d} based on duality relation}\label{sec-conj3d}

Let us define, for any $\alpha>0$ and any Bravais lattice $L\in \mathcal{L}_3$,
$$
F_\alpha[L]:=\sum_{p\in L} e^{- \alpha |p|}.
$$
Using the Laplace transform representation of $r\mapsto e^{-\alpha \sqrt{r}}$, the Jacobi Transformation Formula for the lattice theta function (see e.g. \cite[Prop. 1.12]{BeterminKnuepfer-preprint}) and the change of variable $t=\frac{u\alpha^2}{4}$, we obtain for any $L\in \mathcal{L}_3^\circ(1)$ and any $\alpha>0$, by Fubini's theorem,
\begin{align*}
F_\alpha[L] &=\sum_{p\in L} \int_0^{+\infty} e^{-t|p|^2}\frac{\alpha e^{-\frac{\alpha^2}{4t}}}{2\sqrt{\pi} t^{\frac{3}{2}}}dt=\frac{\alpha}{2 \sqrt{\pi}}\int_0^{+\infty} \theta_L\left( \frac{t}{\pi} \right) \frac{e^{-\frac{\alpha^2}{4t}}}{t^{\frac{3}{2}}}dt=\frac{\alpha \pi}{2}\int_0^{+\infty} \theta_{L^*}\left(\frac{\pi}{t} \right)\frac{e^{-\frac{\alpha^2}{4t}}}{t^{\frac{3}{2}}}dt\\
&=\frac{8\pi}{\alpha^3}\int_0^{+\infty} \theta_{L^*}\left(\frac{4\pi}{u\alpha^2}  \right) e^{-\frac{1}{u}}u^{-3}du.
\end{align*}
Since $u\mapsto e^{ -\frac{1}{u}}u^{-3}$ is decreasing very rapidly and is equal to $0$ for $u=0$ -- i.e. for any $\varepsilon$, $\{u\in \R_+: e^{ -\frac{1}{u}}u^{-3}>\varepsilon\}$ is included in a connected compact set -- $\alpha$ large implies that the minimizer of $L\mapsto F_\alpha[L]$ in $\mathcal{L}_3^\circ(1)$ has the tendency to be the minimizer of $L\mapsto \theta_{L^*}(4\pi/(u\alpha^2))$ in $\mathcal{L}_3^\circ(1)$ where $u$ is in a compact set and $\alpha$ is large. The Sarnak-Str\"ombergsson conjecture \cite[Eq. (43)]{SarStromb} tells us that the minimizer of $L\mapsto \theta_L(\beta)$ on $\mathcal{L}_3^\circ(1)$ is expected to be $\mathsf{D}_3^*$ for $\beta<1$. Therefore, the minimizer of $L\mapsto F_\alpha[L]$ in $\mathcal{L}_3^\circ(1)$ for large values of $\alpha$ is expected to be $\mathsf{D}_3$. By duality, it is clear that the minimizer in $\mathcal{L}_3^\circ(1)$ for small values of $\alpha$ is expected to be $\mathsf{D}_3^*$. This duality relation has been observed by Torquato and Stillinger in \cite[p. 4]{Torquatoduality}.  We numerically compute that $F_\alpha[\mathsf{D}_3^*]<F_\alpha[\mathsf{D}_3] \iff \alpha<3.86$ (see Figure \ref{fig-thetaroot}). 

\begin{figure}
\centering
\includegraphics[width=8cm]{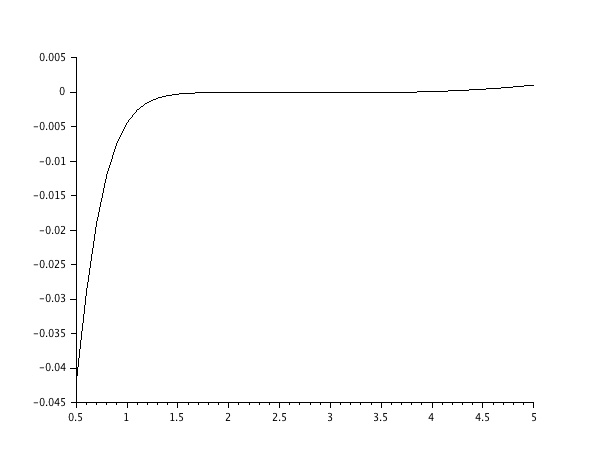} \quad \includegraphics[width=8cm]{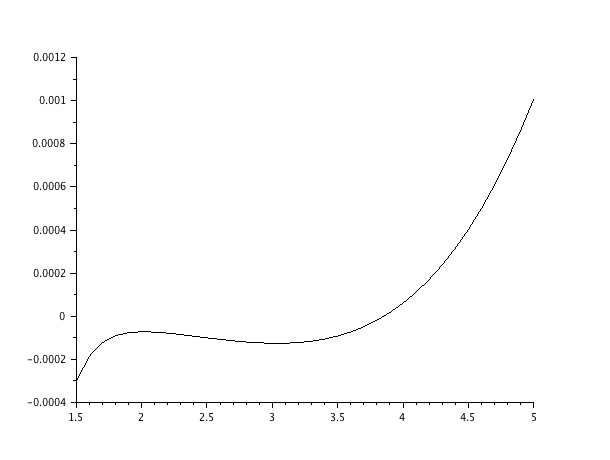}
\caption{Graph of $\alpha \mapsto F_\alpha[\mathsf{D}_3]-F_\alpha[\mathsf{D}_3^*]$ on $[0.5,5]$ (left) and $[1.5,5]$ (right). We observe that we have equality if and only if $\alpha\approx 3.86$.}
\label{fig-thetaroot}
\end{figure}

In particular, the asymptotic minimality as $\alpha\to 0$ and $\alpha\to +\infty$ for the theta function is already known (see e.g. \cite{BeterminPetrache}). We then have the following rigorous result:

\begin{lemma}[Asymptotic minimizer of $F_\alpha$]
As $\alpha\to 0$ (resp. $\alpha\to +\infty$), $\mathsf{D}_3^*$ (resp. $\mathsf{D}_3$) is the unique asymptotic minimizer of $F_\alpha$ in $\mathcal{L}_3^\circ(1)$, i.e. for any $L\in \mathcal{L}_3^\circ(1)$, there exists $\alpha_L$ such that for any $0<\alpha<\alpha_L$, $F_\alpha[L]>F_\alpha[\mathsf{D}_3^*]$ (resp. there exists $\tilde{\alpha}_L$ such that for any $\alpha>\tilde{\alpha}_L$, $F_\alpha[L]>F_\alpha[\mathsf{D}_3]$).
\end{lemma}

Now, for $E_{\alpha,r_0}[L]=e^{\alpha r_0}F_{2\alpha}[L]-2F_\alpha[L]$, it is then expected that the BCC lattice (resp. FCC lattice) is the unique minimizer at high density if $\alpha$ is small enough (resp. $\alpha$ large enough). Furthermore, as in the Lennard-Jones type case, the shape of the minimizer at high density is a good candidate for the global minimization problem in $\mathcal{L}_3$. Explaining why the HCP structure is optimal in $\mathcal{P}_3$ for large $\alpha$ is much more difficult. It certainly follows from the fact that $\theta_{\mathsf{D}_3}(\alpha)<\theta_{\hcp}(\alpha)$ for any $\alpha>0$, as explained in \cite[Ex. 2.6]{BeterminPetrache}, and from the attractive-repulsive form of the potential combined with the integral representation of $F_\alpha[L]$ previously stated. That is why we can expect Conjecture \ref{ConjMorse3d} to be true. 

\begin{remark}[Duality relation]
More precisely, the Poisson Summation formula gives, for any three-dimensional Bravais lattice, since the Fourier transform of $\R^3\ni x\mapsto e^{-|x|}$ is $y\mapsto C(1+y^2)^{-2}$,
$$
E_{\alpha,r_0}[L]=C\alpha^3 \frac{1}{|L|}\sum_{q\in L^*} \left\{\frac{4 e^{\alpha r_0}}{(4\alpha^2 + |q|^2)^2}-\frac{1}{(\alpha^2 + |q|^2)^2}  \right\},
$$
where $C$ is a constant. Therefore, minimizing $E_{\alpha,r_0}$ is equivalent with minimizing 
$$
\widehat{E}_{\alpha,r_0}[L]:=\sum_{q\in L^*} \left\{\frac{4 e^{\alpha r_0}}{(4\alpha^2 + |q|^2)^2}-\frac{1}{(\alpha^2 + |q|^2)^2}  \right\}.
$$
We then observe that, as $\alpha\to 0$, we have $\widehat{E}_{\alpha,r_0}[L]\sim \sum_{q\in L^*} \frac{3}{|q|^4}$ which is expected -- by Sarnak-Str\"ombergsson Conjecture for the Epstein zeta function \cite[Eq. (44)]{SarStromb} -- to be minimized in $\mathcal{L}_3^\circ(V)$, for any $V$, by $L^*=\mathsf{D}_3$, i.e. $L=\mathsf{D}_3^*$. This fact also supports the first part of our Conjecture \ref{ConjMorse3d}.
\end{remark}

\subsection{Comparison of our conjecture with some experimental values of $\alpha$ and $r_0$}\label{sec-experiment}
We now want to compare our Conjecture \ref{ConjMorse3d} with the values of $\alpha$ empirically obtained for metals in \cite{GirifalcoWeizer,LincolnetalMorse,SharmaKachhavaMorse,PamuketalMorse,HungetalMorse} and rare-gas crystals in \cite{Raff1990,AlimietalMorse,ParsonMorse,BarkerMorse}. In order to compare these values, we need a result we have already proved in \cite[Thm 1.11]{OptinonCM} in the Lennard-Jones type potentials case. We recall that the shape of a lattice is its class of equivalence modulo dilation in the fundamental domain of $\mathcal{L}_d^\circ(1)$ where only one copy of each lattice exists (see  \cite[Sec. 1.1]{OptinonCM}).
\begin{lemma}\label{lem:scalingr0}
For any $\alpha>0$ and $r_0>0$, it holds
$$
\displaystyle \argmin_L E_{\alpha,r_0}[L]=\frac{\argmin_L E_{\alpha r_0,1}[L]}{r_0}.
$$
In particular, the shape of the global minimizer for both energies is the same, i.e. its shape is independent of $r_0$
\end{lemma}
\begin{proof}
It is a straightforward consequence of the following equality
$$
E_{\alpha,r_0}[L]=e^{\alpha r_0}\sum_{p\in L}e^{-2\alpha r_0 |p/r_0|}-2\sum_{p\in L} e^{-\alpha r_0 |p/r_0|}=E_{\alpha r_0,1}[L/r_0].
$$
\end{proof}

In \cite{GirifalcoWeizer,LincolnetalMorse,SharmaKachhavaMorse,PamuketalMorse,HungetalMorse}, different values of $\alpha,r_0$ have been computed for metals according to different experimental data. The structure of each metal is taken in its (known) ground-state. More precisely, let us review the results obtained in these papers:
\begin{enumerate}
\item The energy of vaporization, lattice constant and compressibility at zero temperature have been used to compute the parameters in \cite{GirifalcoWeizer}. Therefore, the equation of state and the elastic constants computed with these parameters reasonably agreed with experiment for FCC and BCC metals, as well as all local stability conditions. However, the agreement is more accurate for the FCC metals.

\item In \cite{LincolnetalMorse}, the lattice parameter, bulk modulus and cohesive energy have been used to compute the parameters, cutting off the range of the potential after 176 (for the FCC structure) and 168 neighbours (for the BCC structure). These parameters were used to compute the pressure derivatives of the second-order elastic constants. These computations match fairly well with the experimental values, which is not the case for third-order elastic constants.

\item In \cite{PamuketalMorse}, crystalline state physical properties at any temperature are used to compute the Morse parameters, thus the second order elastic constants are shown to match with the experimental values. The values are reasonably accurate with the one computed in  \cite{GirifalcoWeizer} for the FCC structures, but not so accurate for the BCC one like K and Na. The authors then remarked that ``even though for metals the additive form of the total potential in terms of pair interactions is not a very good approximation, for the sake of simplicity, pair potentials are widely used in calculating various properties of metallic systems." (e.g. for Monte-Carlo-type calculations).

\item The correlation between molecular properties and crystal state have been investigated, by assuming the identity of qualitative behaviour of metals in the two states, in \cite{SharmaKachhavaMorse} for computing the parameters. In particular, this hypothesis permits the invariance of the fundamental potential parameter $\alpha$ in the two states for the metals, where
\begin{equation}\label{def-alphaexp}
\alpha=2\pi w_e\sqrt{\frac{\mu}{2D_e}},
\end{equation}
$w_e$ being the classical frequency data of small vibrations of a diatomic molecule, $\mu$ is the reduced mass of the molecule and $D_e$ represents its dissociation energy. They hence have computed the values of the cohesive energy, thermal expansion, Gr\"uneisen parameter and elastic constants. Satisfactory agreement is obtained for elastic constants of Cu and Pb at zero temperature, which is not true for the thermal expansion and the Gr\"uneisen parameter for the same elements as well as for the cohesive energy. Therefore, computing $\alpha$ from the cohesive energy, the authors found a good match with the others experimental quantities.

\item In \cite{HungetalMorse}, the parameters have been computed, for BCC crystals, using the volume per atom and atomic number in each elementary cell, as well as the energy of sublimation, the compressibility and the lattice constant. These parameters values show a good agreement with the anharmonic interatomic effective potential and the local force constant in X-ray absorption fine structure for Fe, W and Mo.
\end{enumerate}

A first observation is that, in all the cited works, the values of the parameters match with some local quantities (e.g. local deformation of the solid in its ground-state) that are experimentally found from the real (already given) ground-state structure. In this paper, we are interested in the values of the parameters $\alpha,r_0$, thus only for the pairs with $r_0=1$ (by Lemma \ref{lem:scalingr0}), such that the FCC lattice, BCC lattice or HCP structures are the ground-state of the energy per point. Since there is no metal having a HCP structure, we necessarily should have $\alpha r_0<3.5$. In the list of parameters that are computed in the previously cited papers, it only happens for:
\begin{enumerate}
\item K, Na, Cs and Rb in \cite{GirifalcoWeizer}, but there is no match of the real crystal structure with the expected ground-state.
\item K and Na in \cite{LincolnetalMorse}, and their BCC structures match with the real ground-state because the value of $\alpha r_0$ is 3.05643 (we have numerically checked the optimality of the BCC lattice in this case) for K and 2.95443 for Na.
\item Li, Na and K in \cite{PamuketalMorse}, but the expected ground-state structure only matches with real one for Li (BCC) for which $\alpha r_0=2.9751$.
\item Li an Cu in \cite{SharmaKachhavaMorse}, but the expected ground-state structure only matches with the real one for Cu (FCC) for which $\alpha r_0=3.08757$.
\end{enumerate}
Furthermore it never happens for the values given in \cite{HungetalMorse}. Therefore, we can conclude that the central-force model for metals agrees with the minimizing lattice of $E_{\alpha r_0,1}$ only for few BCC structures (Na, K, Li) and only for one element (Cu) that has a FCC structure. It is interesting to remark that these metals atoms with which the structures match are not the heaviest possible -- as we could expect -- thus the approximation of the atoms interaction in metals as a sum of Morse two-body potentials seems to be unsatisfactory, when the parameters $\alpha,r_0$ are chosen according to the real physical properties of the solid.

\medskip

Furthermore, the same comparison can be done with the rare-gas crystals models with Morse potential proposed in \cite{Raff1990,AlimietalMorse,ParsonMorse,BarkerMorse}. For all the values of the parameters, we have $\alpha r_0>3.5$ and the FCC structure that is expected to be the ground-state for all the rare-gas crystals does not match with the HCP structure theoretically expected from our conjecture. This central-force model with Morse pairwise interaction is then again not appropriate for describing the real structure of rare-gas crystals as a ground-state.

%

\medskip

\noindent \textbf{Acknowledgement:} I acknowledge support
from VILLUM FONDEN via the QMATH Centre of Excellence (grant No. 10059). I am also grateful to Xavier Blanc and Jan Philip Solovej for their feedback on the first version of this paper.

{\small
\bibliographystyle{plain} 
\bibliography{BiblioMorse}}
\end{document}